\documentclass[aoas]{imsart}
\usepackage{xr}

\usepackage{graphicx}
\usepackage[utf8]{inputenc}
\usepackage{amsthm}
\RequirePackage{bm}
\usepackage{mathrsfs}
\usepackage{mathtools}
\usepackage{bbm}
\usepackage{subcaption}
\usepackage{float}\label{key}
\usepackage{comment}

\usepackage{arydshln} 

\usepackage[dvipsnames]{xcolor}

\renewcommand{\abstractname}{Abstract}

\newtheorem{proposition}{Proposition}

\newtheorem{definition}{Definition}

\newcommand{\Cb}{\mathbf{C}}

\newcommand{\Gb}{\mathbf{G}}

\newcommand{\Pb}{\mathbf{P}}

\newcommand{\Xb}{\mathbf{X}}
\newcommand{\Yb}{\mathbf{Y}}
\newcommand{\Zb}{\mathbf{Z}}

\newcommand{\bG}{\bm{G}}

\newcommand{\bX}{\bm{X}}

\newcommand{\Prob}{\mathbb{P}}
\newcommand{\Expec}{\mathbb{E}}

\newcommand{\boldtheta}{\bm{\theta}}

\usepackage{tikz}
\usepackage{arydshln}

\usepackage{bbm}

\usepackage[sort,numbers]{natbib}

\newcommand{\btau}{\boldsymbol{\tau  }}

\newcommand{\hv}{\text{hv}}

\newcommand{\KL}{\textbf{KL}}

\newcommand{\jointPxz}{ f (\Zb, \Xb  )}

\newcommand{\Rcz}{ R_{\Cb | \Zb }  }
\newcommand{\btheta}{\boldsymbol{\theta}}

\usepackage{color, colortbl}
\definecolor{Gray}{gray}{0.9}
\definecolor{LightCyan}{rgb}{0.88,1,1}
\definecolor{mygrey}{gray}{0.45}

\usepackage{todonotes}
\usepackage{url}

\newcommand{\ch}[1]{\textcolor[rgb]{0,0,0}{{#1}}}
\newcommand{\toch}[1]{\textcolor[rgb]{0,0,0}{{#1}}}

\usepackage{enumitem}
\usepackage{rotating}
\usepackage{makecell}

\newcommand{\ERhv }{ \Expec_{ R (\mathcal Z, \Cb  )} }    

\usepackage{sectsty}
\usepackage{titlesec}

\RequirePackage{amsthm,amsmath,amsfonts,amssymb}

\makeatletter
\newcommand*{\addFileDependency}[1]{
	\typeout{(#1)}
	\@addtofilelist{#1}
	\IfFileExists{#1}{}{\typeout{No file #1.}}
}
\makeatother

\newcommand*{\myexternaldocument}[1]{%
	\externaldocument{#1}%
	\addFileDependency{#1.tex}%
	\addFileDependency{#1.aux}%
}

\myexternaldocument{supplement_RENEW}

\begin{document}
	
	\begin{frontmatter}
		\title{ Community Detection in Weighted Multilayer Networks with Ambient Noise }

		%
		%
		%
		
		\begin{aug}
			\author[A]{\fnms{Mark} \snm{He}\ead[label=e1]{markhe@live.unc.edu}},
			\author[B]{\fnms{Dylan (Shiting)} \snm{Lu}\ead[label=e2]{shl16101@live.unc.edu}},
			\author[B]{\fnms{Rose Mary} \snm{Xavier}\ead[label=e2,mark]{rxavier@email.unc.edu (joint senior author)}}
			\and
			\author[C]{\fnms{Jason} \snm{Xu}\ead[label=e3,mark]{jason.q.xu@duke.edu,  (joint senior author)}}
			\address[A]{Herbert Irving Comprehensive Cancer Center, Columbia University, New York, NY, 11032, USA	\printead{e1}}
			\address[B]{School of Nursing, University of North Carolina, Chapel Hill, NC, 27599, USA
				\printead{e2}}
			\address[C]{ Statistical Science, Duke University, Durham, NC, 27708, USA
				\printead{e3}}
		\end{aug}
		
		%
		
		\begin{abstract}
			We introduce a novel model for multilayer weighted networks that  accounts for global noise in addition to local signals. The model is similar to a multilayer stochastic blockmodel (SBM), but the key difference is that between-block interactions independent across layers are common for the whole system, which we call \textit{ambient noise}. A single block is also characterized by these fixed ambient parameters to represent members that do not belong anywhere else. 
			\ch{This approach allows simultaneous clustering and typologizing of blocks into \textit{signal} or \textit{noise} in order to better understand their roles in the overall system, which is not accounted for by existing Blockmodels.} 
			We employ a novel application of hierarchical variational inference to jointly detect and \ch{differentiate types of blocks}. We call this model for multilayer weighted networks the \textit{Stochastic Block (with) Ambient Noise Model }(\texttt{SBANM}) and develop an associated community detection algorithm.  We apply this method to subjects in the Philadelphia Neurodevelopmental Cohort to discover communities of subjects with co-occurrent psychopathologies in relation to psychosis.
		\end{abstract}
		
		
	\end{frontmatter}
	
	
	\section{Introduction}
	
	In the advent of more sophisticated data gathering mechanisms and more nuanced conceptions of  dependency,
	relational data  have become more widely used than ever before. Statistical network analysis has become a major field of research and is a useful, efficient mode of  pattern discovery. Networks representing  social interactions,  genes, and ecological webs often model members or agents as nodes (vertices) and their interaction as edges.
	\ch{Oftentimes,  relational information  for the same observations or participants manifest in different modalities, represented by layers}. For example, nodes represented by users in a social network such as Twitter can have edges that represent `likes', `follows', and `mentions'.
	In biological  networks,  modes of interactions such as gene co-expressions or similarities between biotic assemblages may arise among the same sample of study. 
	The study of multilayer networks is especially pertinent in the study of psychiatric data, where distinct diagnoses are not clearly demarcated but rely on constellations of interacting psychopathologies. In this study, we analyze these multimodal psychopathological symptom data using multilayer network analysis.
	
	While network theory for simple graphs is well established \cite{girvan_community_2002,configbender}, the literature concerning weighted, multimodal networks is an emerging field of interest \cite{Menichetti_2014WeightedMultiplex,Holme_2015}.
	The field of \textit{community detection} has also grown considerably in recent times \cite{newman2018networks,Fortunato_2016,handock_modelclust_2007,townshend_murphy_2013}.
	Community detection is an approach used to cluster nodes in a network. Many techniques have been proposed for \emph{unweighted} (binary) graphs including modularity optimization \cite{girvan_community_2002, clauset_community_2004}, stochastic block models \cite{holland_stochastic_1983, nowicki_snijders, peixoto_nonpar, yan_sbm}, and extraction \cite{zhao_consistency_2012,Lancichinetti_plos_2011}.  
	
	The stochastic block model (SBM) is a theoretical model for  random graphs \cite{karrer2011stochastic,peixoto_nonpar,Hoff2002,nowicki_snijders}; it has also found practical use in community detection  \cite{mariadassou2010,newman2003structure}.  The model lays out a concise formulation for dependency structures within and across communities, but does not model \textit{global} characteristics. 
	Though some methods  discern \textit{background} (unclustered) nodes \cite{palowitch_continuous,wilson_essc,dewaskar2020finding}, few existing models explicitly account for \textit{community-wise} noise even though it may be useful. We develop a  model  for  multilayer weighted graphs that explicitly accounts for (1)  global noise present between differing communities, and (2) dependency structure across layers within communities.  We call this model and its associated estimation algorithm as the  (multivariate Gaussian) \textit{Stochastic Block (with) Ambient Noise  Model} (\texttt{SBANM}).
	
	We develop a novel method that jointly finds clusters in a multilayer weighted  network  and classifies the \textit{types} of these clusters, namely whether they are (local) signal or (global) noise. We propose a model that discovers and categorizes these communities.
	We also develop its method of inference, which is additionally useful as many existing multilayer SBM analyses  assume known parameters \cite{wang2019joint,mayya2019mutual}.  In the primary case study (Section \ref{sec:results_PNC}), we use \texttt{SBANM} to find clusters of diagnostic subgroups of patients judged by similarity measures of their psychopathology symptoms.

	\subsection{Background and Contributions} \label{sec:contributions_context} 
	
	A canonical example of a globally noisy network is the Erdos-Renyi model where every edge is governed by a single probability. The affiliation model is a weighted extension \cite{Allman_2011} used to describe a ``noisy homogeneous network"; a single \textit{global} parameter $\theta_\text{in}$ dictates the connectivity between all nodes in \textit{any} community, while another  $\theta_\text{out}$ controls the connectivity for all nodes in differing communities.  A similar model was posited by Arroyo et al. \cite{arroyo2020inference}  where $\theta_\text{in} > \theta_\text{out}$ as a baseline for network classification. The weighted SBM and the affiliation model are both \textit{mixture models for random graphs} described  by Allman et al. \cite{Allman_2011,ambroise2010new}. 
	
	SBMs were initially used for simple networks  \cite{Hoff2002,nowicki_snijders}, but they have been extended to weighted \cite{mariadassou2010}) and multilayer settings \cite{stanley2015,paul_multilayer_2015,arroyo2020inference}, and in particular  time series \cite{matias2016} where clusters across all time points have the same inter-block parameters, but varying between-block interactions.  These multilayer SBMs typically do not  account for correlations between layers.  Some recent studies or binary networks  have accounted for correlations across layers \cite{mayya2019mutual} and noise \cite{mathews2019gaussian}, but typically assume that parameters are already known.

	Though much work has been done on estimating SBMs, \ch{not much of it has focused on assessing the \textit{noise} present within them}, much less for multilayer weighted graphs. Extraction-based methods identify background nodes to signify lack of community membership \cite{palowitch_continuous,wilson_essc}, but these methods do not attribute any parametric descriptions to these nodes. 
	Some recent work  discuss noise in network models that are  oftentimes associated with global (i.e. entire-network) uncertainty \cite{blevins2021variability,Newman_2018,mathews2019gaussian,Young_2020}. However, few have studied \textit{structural noise} that exists between differing communities or that serves as some notion of a residual term (i.e. in regression analysis).

	We attempt to address these two gaps in this work. In a multilayer graph with $Q$ ground truth communities (indexed by $q $), as well as a single block that is considered \textit{noise} (labeled $NB$  for \textit{noise block}), we postulate a model  that is \textit{locally unique} with parameter $  \btheta_q  $  for all edges within a block indexed at $q$. The global noise parameter ${\btheta} _\text{Noise} $ describes all interactions between differing blocks as well as $NB$. 
	A simplified version of this model is presented  as follows, but will be written in more detail in Section \ref{sec:model_inference}:
	\begin{align} \label{eq:combined_model}
		\btheta_{ql}    & = \begin{cases}
			{\btheta}_{q} & \text{if } q = l \text{ and } q \text{ is not } NB
			\\
			{\btheta} _\text{Noise} & \text{if } q \ne l \text{ or } q \text{ is } NB    
		\end{cases}.
	\end{align}  
	
	The model combines qualities of the affiliation model \cite{Allman_2011} with the weighted SBM and extends to multiple layers.  Because both the affiliation model and the multilayer SBM are proven to be identifiable by prior work \cite{Allman_2011,matias2016}, we posit that \texttt{SBANM} is also identifiable. A brief argument is given in Appendix \ref{app:identifiability}, but deeper investigation remains as future work.  One major advantage of a global noise term is its parsimony  compared to SBMs.  Existing clustering models on multilayer networks, even when accounting for communities that persist across layers \cite{persistent_liu}, still tend toward overparameterization. 
	A reference or \textit{null} group is often used in scientific and clinical settings, an example being the cerebellum in the analysis of brain networks.  

	\subsection{Motivation} \label{sec:motivation}  
	We use an example to motivate the proposed model. Suppose there is a social network where nodes represent individuals and weights represent social interactions among them. Individuals naturally interact in cliques where rates of communication are roughly similar (i.e. assortative).  Across differing communities, however, rates are assumed to be at a global baseline level. Moreover, interactions among  members who are \textit{asocial}  and do not belong to any community with a unique signal are similarly modeled as ``noise".  Which individuals are still friends with each other after 10 years?  Alternatively, how might relationships be broken down -- in what ways may work relationships (i.e. co-authorships) correlate with friendships?  A schematic figure for this model compared to SBM is presented in Figure \ref{fig:schematic}.
	
	\begin{figure} [htbp]
		\centering
		\begin{tabular}{cc}
			\hline
			\textbf{	SBM }&\textbf{\texttt{SBANM} }\\
			\hline
			\includegraphics[width=0.35\linewidth, trim= {3cm 19.5cm 12.5cm 2cm }, clip]{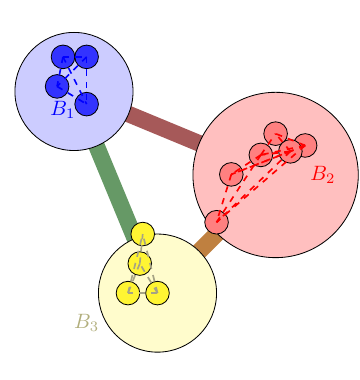} 	&  
			\includegraphics[width=0.5\linewidth, trim= {2.5cm 16cm 8.5cm 2cm }, clip]{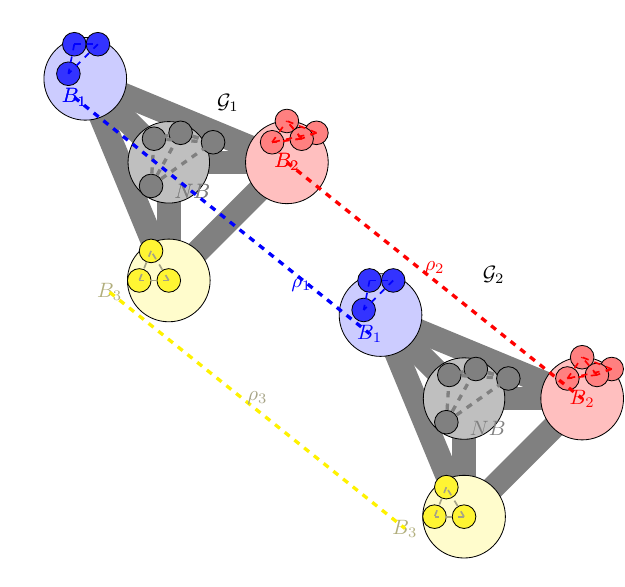} 
		\end{tabular}
		\caption{ \footnotesize 
			\textit{Illustrative example of the types of relationships between  blocks for the canonical SBM (left)  and \texttt{SBANM} (right). Dashed lines represent the inter-block connectivity among nodes. Large circles represent distinct communities. Solid thick lines represent the inter-community rates of interaction (transition probabilities if binary).  In the canonical case (left), the inter-block transitions are all distinct, as denoted by its colors. For the multilayer \texttt{SBANM} case (right), the  inter-block parameters are all the same (represented by gray). \textit{Ambient Noise} ($AN$) governs the connectivities between blocks and the \textit{intra-}block connectivity within the block \textit{NB} across two layers $\mathcal{G}_1$ and $\mathcal{G}_2$ with blocks $B_1,B_2,B_3$ and \textit{Noise Block} ($NB$) with correlations $\rho_1,\rho_2,\rho_3 $ across $\mathcal{G}_1$ and $\mathcal{G}_2$ .	}}
		\label{fig:schematic}
	\end{figure}
	
	Psychiatric disorders lack objective measures such as laboratory testing that can confirm or clarify diagnosis.
	As such, the diagnostic process rests on clinical assessment and is built on codified symptom domains  \cite{kahn_schizophrenia_2015}.	Psychiatric illnesses, moreover, have multiple causes and symptoms.  There are no laboratory tests for most psychiatric conditions;  current diagnostic processes only consider the presence of discrete symptoms and can identify patients who have already have the disease, but it does not help identify who is at risk for the illness in question.  One such illness is schizophrenia, a chronic psychotic disorder that affects millions worldwide and imposes a substantial societal burden. Identifying individuals who are at risk for developing this illness is an important issue.

	\ch{In most existing research on networks, nodes represent individuals and edges are known quantities between nodes}. This assumption cannot directly be applied to psychiatric network models to identify communities of individuals with specific conditions, \ch{as such relational data are not measurable}. They can, however, can be estimated from biological and/or psychosocial data, which can then be used for early identification \cite{kahn_schizophrenia_2015,clark_dx95,kendell_diagnoses2003}. \toch{The flagship criterion that defines the diagnosis of disorders such as schizophrenia is the presence of positive symptoms (DSM-V \cite{APA2013}). This type of categorization  is clinically useful but leads to an excess of diagnostic comorbidities and heterogeneities in the clinical presentation of illnesses \cite{clark_dx95,kendell_diagnoses2003,nosology}.  More importantly, it is a post-hoc diagnosis: subjects typically are no longer treatable after being diagnosed.} 
	With an increase in availability of multimodal data across populations of clinical subjects,  community detection is a natural tool for the classification of psychiatric illnesses with multifaceted latent characteristics that could not be directly observed. \ch{Moreover, it can pave way for future methods to contribute towards the important objective of early identification.}
	
	\toch{We use the ``co-occurrences" in \textit{psychopathology} symptoms to detect groups of participants with early signs of psychosis.} \ch{Existing research document ``co-occurring and reciprocal relationships between" anxiety,  mood, and behavior disorders among cohorts that form a distinct \textit{prodromal} stage that precedes psychosis \cite{Cupo2021,Krabbendam2005DevelopmentOD}. 
		We model the \textit{co-occurrence} among common prodromal symptoms as the correlations between networks constructed from anxiety, behavior, and mood  disorders. These networks are highly correlated in the prodromal stage, but become independent once the threshold of psychosis is crossed. 
		After the initial episodes of psychosis, the diseases progresses qualitatively out of the \textit{prodrome} and into a psychotic illness. \cite{Tandon2012AttenuatedPA}.
		The \textit{independent} group should manifest high levels of psychosis and represent the group that have transitioned from psychosis prodromal symptoms to active psychotic symptoms \cite{Tandon2012AttenuatedPA}, while the \textit{correlated} groups of subjects remain in the prodrome.
	}
	
	\ch{
		We hypothesize that the independence assumption of the \textit{NB} cluster from the \texttt{SBANM} model corresponds to the  decoupled prodromal symptoms in the psychosis symptoms stage.  We statistically model the separation in the stages of psychosis symptom onset with the \texttt{SBANM} model and algorithm.   The model seeks to separate  groups that have transitioned to active psychosis symptoms or rather psychosis spectrum from those that have not. We view the correlations between networks constructed from anxiety, behavior, and mood disorders as the analogous to the ``co-occurrence and reciprocal relationships" among the prodromal symptoms. 
		Consequently, we interpret clusters with high correlations among these pathologies as indicative of the subjects in varying prodromal stages of psychosis development.  In contrast, we hypothesize that the subjects found within \textit{NB}, whose psychopathologies are independent across network layers, are indicative of subjects that have converted to the psychosis spectrum stage. This analysis will be described in more detail in Section \ref{sec:app_results}.
	}

	In Section \ref{sec:data}, we describe the terminology alongside the \textit{Philadephia Neurodevelopmental Cohort} (PNC) data for the main case study. We then describe the model and its method of (variational) inference in Section \ref{sec:model_inference}, and its specific mechanics in Section \ref{sec:est_algorithm}.  In Section \ref{sec:results_PNC}, we describe the analysis and results for the PNC case study.
	In Section 6, Model performance is assessed using synthetic experiments that closely match the results derived from the data. 
	While distinguishing psychosis spectrum will be the primary focus of the proposed methodology, it is useful to find latent structure in other types networks.  We also demonstrate the method on (1) US congressional voting data and (2) human mobility (bikeshare) data in Appendices \ref{app:voting},\ref{app:results_bikeshare}.

	\section{Data, Notation, and Terminology} \label{sec:data}
	

	For  a $K$-layer \textbf{weighted} multigraph with registered $n$ nodes indexed by the set $[n] = \{1,2,...,n  \}$, let $\Gb$ represent the collection of multilayer weighted graphs with $K$ layers: $\Gb = \{  \Gb^1, \Gb^2,...,\Gb^K\}.$
	Similarly, suppose $\Gb$ contains $Q$ ground truth communities (blocks) indexed by $q$, but such that a single block is called \textit{noise block} and labeled $NB$ (indexed by $q_{NB}$). 
	We let  $\textbf{G}_{ij} = (G^1_{ij} , G^2_{ij},..., G^K_{ij})$ represent the vector of edge-weights between edges $(i,j)$ across all layers $k=1,2,...,K$. We define a community  as  $B_q \subset [n]$ to denote the nodes that are contained in a given block indexed by $q$  in $\Gb$, and we let $\Gb_q$ represent the set of all edges contained in block $q$ across all $K$ layers:
	\begin{align} \label{eq:Xb_q} 
		\Gb_q = 	\{  \Gb_{ij}    \}_{i,j \in B_q}.
	\end{align}  
	Moreover, we call the set of edges across different blocks $q,l$ (where $q \ne l$) \textit{interstitial noise} ($IN$), and label it as:   		
	\begin{align} \label{eq:Xb_IN}
		\Gb_{IN} = 	\{  \Gb _{ij}   \}_{i  \in B_q , j \in B_l }. 
	\end{align} 
	We fix \textbf{one} block indexed  as $NB$ as the \textit{noise block}, where all weights in the block follow a $ N_K \big(\boldsymbol{\mu}_{NB}, \boldsymbol{\Sigma}_{NB} \big)  $ distribution. 	This block represents a null region that is devoid of unique signal, but is distributionally governed by the same characteristics as the interstitial relationships between different blocks.
	We let $\Gb_{NB}$  represent the set of edges among members in the ``noise block": $\Gb_{NB} = 	\{  \Gb_{ij}   \}_{i,j \in NB}. $
	In the following subsection we describe the data as introduced in the prior section in the context of the notation.
	In Section \ref{sec:model_inference} we describe the assumption that classifies this notion of noise. 
	Multilayer networks can represent  multimodal,  longitudinal, or \textit{difference} graphs
	\cite{Menichetti_2014WeightedMultiplex,Holme_2015}.  The  data in   the  \textit{Philadelphia Neurodevelopmental Cohort} (PNC) (described below) is constructed as a multimodal network, while the applications outlined in Appendices \ref{app:voting} and \ref{app:results_bikeshare} are examples of longitudinal graphs.
	$\Gb =\{\Xb, \Yb, \Zb \}$ represents  anxiety, behavior, and mood psychopathology symptom networks processed from the psychopathological questionnaires for a given set of subjects. 
	With respect to the PNC data, each layer represents one of the psychometric evaluation networks for each disorder.

	\ch{
		In the introduction we have mentioned that sex differences play a large role in psychosis onset \cite{Mc1988GenderDI,Nowrouzi2015AgeAO,Kirkbride_2012_Sexdiff,Li_SexDiff_2016}. The prevalent view is that males typically have a peak in the rates of onset between the ages of 21-25, while females have bimodal peaks much later \cite{Li_SexDiff_2016}. We subset the PNC data to \textit{early adult} (ages 18-21) males for the statistical analysis to be concordant with the clinical context. 
	}
	The sample size $n$ in  this study represents the 764 \textit{early adult} male subjects. 
	Each node represents a subject, and each weighted edge the transformed similarity ratio between two subjects for anxiety, behavior, and mood symptoms.
	
	The Philadelphia Neurodevelopmental Cohort (PNC) is a community sample of youth subjects aged 8-21 years, recruited from the greater Philadelphia area. These subjects underwent  a detailed neuropsychiatric evaluation. \cite{Calkins_PNC_2014,Calkins_PNC_2015}. A PNC subsample is used as the primary case study.
	We assume each member of $\Gb $ are generated from node-clusters whose (Fisher) transformed edges  follow blockwise multivariate normal distributions.  We use three general categories of disorders to represent each layer:
	\begin{enumerate}[noitemsep]
		\item Anxiety ($\Xb$): 44 questions (generalized anxiety, social anxiety, separation anxiety, agoraphobia, specific phobia, panic, obsessive compulsive and post traumatic stress disorder)
		\item Behavior ($\Yb$): 22 questions (attention deficit hyperactive disorder, oppositional defiant disorder, conduct disorder) 
		\item Mood ($\Zb$): 11 questions (depression and mania),
	\end{enumerate} 
	then Fisher-transformed to produce the weighted edge in graph layer   In these following sections these categories will simply be referred to as ``anxiety",``behavior", and ``mood". More details on pre-processing can be found in Appendix \ref{app:data_pnc_proc}.
	
	\ch{
		We note that $NB$ could in some cases parallel the notion of an \textit{indepedent residual}, but not always. In case study of PNC data presented in Section \ref{sec:results_PNC}, $NB$ corresponds to perhaps the ``most informative" discovered block.  In Section \ref{sec:motivation} 
		we hypothesize that the independence of $NB$ across layers suggests separation of the prodromal co-occurrences 
		As such, we say that $NB$ is \textit{noise} only insofar as it is independent across layers, paralleling the analogy of the residual in regression analysis, when in practice $NB$ may correspond to the cluster of the most interest.
	}

	\section{Model and Inference} \label{sec:model_inference} 
	
	\texttt{SBANM} supposes that networks across $K$ layers have the same block structure, while transition parameters between blocks are fixed at the same global level.  This model  allows detection of common latent characteristics across layers, as well as differential sub-characteristics within blocks (represented by multivariate normal distributions).  We assume edges  are correlated across layers in the block structures of the proposed model.  
	

	\begin{definition} \label{def:corr_blocks} 
		{(Correlated Blocks)} 
		\normalfont		
		For a $K-$layer (Gaussian) weighted multigraph $\Gb=\{\Gb^1,...,\Gb^K\}$ where each layer $k$ represents a graph with  $n$ registered nodes, let $B_q \subset [n]$ represent a community housing a partition of  nodes $\{i\}_{i \in B_q}$, then each weighted edge between any node in block $B_q$ form a multivariate normal distribution with mean $K$-dimensional vector $ \boldsymbol{\mu}_q $ and $K \times K $-dimensional covariance matrix $\boldsymbol \Sigma_q$:
		\begin{align*}  
			\boldsymbol{\Sigma}_q  = 
			\begin{pmatrix}
				\sigma_{q,1}^2& \rho_q \sigma_{q,1} \sigma_{q,2  }... & \rho_q \sigma_{q,1} \sigma_{q,K }   \\
				\rho_q \sigma_{q,2} \sigma_{q,1 }  &   \sigma_{q,2}^2... & \rho_q \sigma_{q,2} \sigma_{q,K }  \\
				... &... &...  \\
				\rho_q \sigma_{q,K} \sigma_{q,1 }  &  ... & \sigma^2_{q, K} 		
			\end{pmatrix}.
		\end{align*}   	
		If nodes $i,j$ are in the same block, the distribution of their edges follow a multivariate normal distribution
		\begin{align*}
			\Gb_{ij} |   \{ i \in B_q, j \in B_q \}  \sim N_K (\boldsymbol{\mu}_q , \boldsymbol{\Sigma}_q).
		\end{align*}
	\end{definition} 
	
	Note that there is a single correlation parameter $\rho_q$ across all layers for a given block $B_q$. This is a deliberate choice to induce parsimony and interpretability among block relationships across all layers. We  assume that the \textit{noise block}  has the same characteristics as the \textit{interstitial noise}; both are drawn from the same distribution $AN$ (\textit{ambient noise}). 	$AN$ is a global noise distribution that governs both  $IN$ and $NB$:
	\begin{align*}
		\Gb_{IN}
		&  \stackrel{d}{=}   \Gb_{NB}
		\sim N_K    (\boldsymbol{\mu }_{AN},    \boldsymbol{\Sigma}_{AN}).
	\end{align*}	
	Because $NB$ and $IN$ both represent ``baseline" levels of connectivity for the network, we assume that they both have equivalent characteristics as $AN$.  Members of each block $B_q$ interact with other members in the same block at  rates that follow multivariate   $\boldsymbol{\mu}_{q}$ with variance $\boldsymbol{\Sigma}{q}$, but interact with members in differing groups $l; l \ne q$ at baseline rates $\boldsymbol{\mu}_{IN}$ with variance  $\boldsymbol{\Sigma}_{IN}$, i.e.   background  interactions. 
	
	\begin{definition} \label{def:ambient_noise}
		{(Ambient Noise)}  \normalfont  Edges in $IN$ between differing blocks and  in $NB$,  are characterized by  
		$ (\boldsymbol{\mu}_{AN} , \boldsymbol{\Sigma}_{AN} ) $:  
		$ \boldsymbol{\Sigma}_{AN}  $ is a $K \times K$ diagonal matrix 
		with diagonal    $ (\sigma^{2}_{AN, 1}, ..., \sigma^2_{AN, K}) $
		and off-diagonal entries of 0:
		\begin{align*}
			\Gb_{ij} |   \{ i \in B_q, j \in B_l \}  \sim N_K (\boldsymbol{\mu}_{AN} , \boldsymbol{\Sigma}_{AN}).
		\end{align*}
	\end{definition}
	For a community   $B_q \subset [n]$ representing the nodes that are contained in block $q$ in a weighted multilayer network $\Gb$, we let $\Gb_q$ represent the set of all edges contained in block $B_q$ across all $K$ layers as defined in Equation \eqref{eq:Xb_q}. Conversely, the set of edges across differing $B_q,B_l$ (i.e. interstitial noise), are defined as in Equation \eqref{eq:Xb_IN}.

	\begin{definition}  
		(Stochastic Block (with) Ambient Noise Model (\texttt{SBANM})) \normalfont
		A $K-$layer (Gaussian) weighted multigraph $\Gb=\{\Gb^1,...,\Gb^K\}$   with $n$ nodes (with index set $[n]$) and $Q$  communities (blocks) indexed by $q$ with a single block that is considered \textit{noise} labeled $NB$  (indexed by $q_{NB}$) with disjoint blocks  $\{B_1, B_2, ...,NB,...,B_Q\}_{q:  q \le Q}$ 
		is a  \texttt{SBANM}  if the following conditions are satisfied.		 
		\begin{enumerate}
			
			\item Edges $\Gb_{ij}$ in the same \textit{correlated} block $B_q$    follows conditional distribution   $N_K(\boldsymbol{\mu}_q , \boldsymbol{\Sigma}_q)$ given block memberships. \ch{Each two edges in the same block are correlated at rate $\rho_q$, across any 2 layers.}
			
			\item Ambient noise $AN$ with $	 N_K    (\boldsymbol{\mu }_{AN},    \boldsymbol{\Sigma}_{AN})$  governs both $IN$ and $NB$: 
			\begin{enumerate} 
				\item Edges $i \in B_q $ and $j \in B_l$ $ (l \ne q)$  follow a $N_K(\boldsymbol{\mu}_{AN} , \boldsymbol{\Sigma}_{AN} )$ distribution. 
				\item   \textbf{One} block $NB$ contains members whose edges are generated from a $K-$ dimensional  multivariate normal distribution $  N_K (\boldsymbol{  \mu  }_{AN},   \boldsymbol{  \Sigma  }_{AN}   ).   $ 
			\end{enumerate} 	
		\end{enumerate}
	\end{definition}
	
	\subsection {Connection to Existing Models}
	
	The weighted SBM and the affiliation model are both cases of the \textit{mixture models for random graphs} described  by Allman et al. \cite{Allman_2011,ambroise2010new}. This general class of network models accounts for assortativity (the tendency for nodes who connect to each other at similar intensities to cluster together) and sparsity (when there are much fewer edges than nodes).
	In addition to VEM-based inference methods \cite{mariadassou2010,matias2016,paul_randomfx_2018} that are extensively referenced in Section \ref{sec:contributions_context}, \ch{we also note the existing multilayer work in physics \cite{Peixoto2017ModellingSA, Barbillon2015StochasticBM,Taylor2016,Valles_2016} and statistics \cite{Lei2022BiasadjustedSC,Macdonald2020LatentSM,wang2019joint}}.

	Some methods for multigraph SBMs are based on spectral decomposition \cite{wang2019joint,arroyo2020inference,mayya2019mutual}. These methods are typically applied to binary networks and use different sets of methodology or assumptions such as known parameters \cite{mayya2019mutual}, but are still similar enough to warrant comparison.   \ch{The notion of \textit{ambient noise} has been studied in some existing methods. Miao et al. \cite{Miao2021InformativeCI} and Priebe et al. \cite{Priebe_2019} present another spectral method for  the goal of \textit{core identification} and  to separate cores (i.e. signal) from periphery (i.e. noise).
		Zhang et al. detect noise using correlations \cite{Zhang_Noise_2015}}.
	
	One class of these existing methods model edge connectivity of a (potentially multilayer) network as a function of  membership vectors $\Zb_i$ (for node $i$), connectivity matrix $\mathcal{R}_k$ at layer $k$, and the graph Laplacian  \cite{mayya2019mutual,mathews2019gaussian,reeves2019geometry,arroyo2020inference,wang2019joint}.  Typically, the connectivity rate corresponds to  Bernoulli probabilities (for binary networks), but some of these approaches allow for extensions to the weighted cases \cite{wang2019joint,mercado2019spectral,arroyo2020inference}.
	Some work has focused on studying the correlations or linear combinations of the eigenvectors of $\mathcal{R}_k$,  but in most of these cases \textit{conditional independence given labels} between layers is assumed \cite{mayya2019mutual,arroyo2020inference}.
	Another class of these multiplex methods is to devise an optimal aggregation scheme to combine multiple layers and then to use single-graph methods on the resultant  network \cite{levin2019central}. We consider several special cases of \texttt{SBANM} that reduce to existing models.  
	\begin{enumerate} [noitemsep]
		\item If all $\rho_q$ were zero (ie. diagonal $\boldsymbol{  \Sigma}_q$; no correlations amongst communities) and all the within-community signals were the same, then   \texttt{SBANM} is a multivariate extension of the  models posited by Allman et al.  \cite{Allman_2011} or Arroyo et al. \cite{Arroyo2019connectomics}.
		\item If $K=1$, \texttt{SBANM} is a special case of the weighted Gaussian SBM as proposed by Mariadassou et al.  where all inter-block connectivities are fixed at a single rate \cite{mariadassou2010}.  
		\item 
		Wang et al. (\cite{wang2019joint}) constrain the connectivity matrix to a diagonal, which would be analogous to \texttt{SBANM} if ambient noise parameter is fixed at zero: $\boldtheta_{AN}:=0$.
		\item
		Arroyo et al. \cite{arroyo2020inference} describe the  multilayer SBM  \cite{holland_stochastic_1983} for binary graphs which   ``could be easily extended to the weighted cases". The model  assumes   \textit{independent} block parameters $\mathcal{R}_k$ across every layer.
		If there were parameters $  \boldtheta_{AN}$ such that $\mathcal{R}_{ql,k}:= \boldtheta_{AN}$ (for every $q \ne l$), then a special case of  \texttt{SBANM} (where each $\rho_q:=0$) would be recovered.
	\end{enumerate}
	
	\ch{One could conceive of many different other alternative models without some of the assumptions about ambient noise, such as, for example, a noise block that does not require between-block parameters to be the same. Indeed, there can be many nested variations, but we choose this specific model because it is parsimonious, applicable to the primary case study of the \textit{Philadelphia Developmental Cohort}, and can generalize to other potential uses.}

	\ch{ \subsection{Hierarchical ELBO}  \label{sec:ELBO_Optim}} 
	\ch{ We estimate our proposed model using  variational inference (VI)} which is used to estimate SBM memberships as well as their parameters \cite{mariadassou2010,paul_multilayer_2015}. 
	VI is an approach to approximate a conditional density of latent variables using optimization \cite{blei_VI,Jaakkola00tutorialon}. 
	When optimizing the full likelihood is intractable, simpler surrogates of complicated variables are chosen to create a simpler objective function. 
	The Kullback-Liebler (KL) Divergence between this simpler function  and the full likelihood are then minimized. For community detection problems, mean-field (MF) approximations of membership allocations often serve as simpler surrogates of latent approximands to simplify the likelihood function into a {lower bound} (typically known as \textit{evidence lower bound}: ELBO) \cite{mariadassou2010,townshend_murphy_2013,airoldi2007mixed}. 
	\ch{Variational inference  is often used for community detection \cite{sarkar_ambient_JMLR,Yin2020ATC,zhou_fan_variational2020,bickel_variational_13}}.

	\begin{definition} {(Evidence Lower Bound (ELBO))}  \label{def:ELBO}
		For observed data $\mathcal X$  with unknown latent membership variables $\mathcal Z $, 	the evidence lower bound (ELBO) 
		$ \mathcal{L} $ is the 	approximately  optimal likelihood that minimizes the KL Divergence between the approximate distribution $R( \mathcal Z , \Cb )$  and the  posterior frequency $f(   \mathcal Z , \Cb |   \mathcal X ) $. It is expressed as follows:	
		\begin{align*} 
			\mathcal{L}  
			&= \Expec_{ R_{\text{hv}} (\mathcal Z)  }    \left[     
			\log f(\mathcal Z, \mathcal X)
			- \log R_\text{\hv} (\mathcal Z ) \right] 
		\end{align*} 
	\end{definition}  
	This ELBO is minimized in variational inference problems. 
	Ranganath et al \cite{ranganath2016}  have shown that the \textit{hierarchical ELBO} yields a tighter bound than the ELBO as
	defined above.
	$$ \mathcal{L}^\prime  = \Expec_{ R (\mathcal Z, \Cb  )}[     	\log f(\mathcal Z, \mathcal X)]     - \Expec_{R ( \mathcal Z , \Cb ) }\left[    \log R ( \mathcal Z, \Cb) \right]   + \Expec_{R ( \mathcal Z, \Cb ) }  \left[  \log S( \Cb |  \mathcal Z) \right].$$  
	An inequality is shown between  the ``ordinary" ELBO and the Hierarchical ELBO by Ranganath et al. \cite{ranganath2016} when minimized with parameters $\Theta$ (defined in the following section)
	(details in Appendix \ref{app:ELBO_Proof}) 
	\begin{align*}
		\min_\Theta 	\mathcal{L}^\prime  \le \min_{\Theta}\mathcal L.
	\end{align*} 

	\subsection{Variational EM}
	Variational EM (VEM) is a demonstrably effective method to estimate SBM and more efficient than other approaches (such as MCMC) \cite{mariadassou2010,nowicki_snijders}. Daudin et al. introduced using VEM for binary SBMs (\cite{Daudin2008}. 
	Mariadassou et al. used a similar method  for  weighted graphs \cite{mariadassou2010}.  
	Though it enables efficient inference, MF VI is limited by its assumption of strong factorization and does not capture posterior dependencies between latent variables arising amongst multilayered networks. Hierarchical variational inference (HVI) provides a natural framework for the two-layered latent structure for multilayer networks. 
	A hierarchy is induced in \texttt{SBANM}  by the assumption that all but one block are categorized as \textit{signal}, while a single block is designated as noise. HVI augments variational approximations with priors on its parameters: this assumption allows joint clustering of blocks and their signal-noise differentiation.
	
	We use a similar approach to that originally used in Daudin et al. \cite{Daudin2008}. 
	The latent variable of interest is the membership allocation matrix $\mathcal Z$, which is a $n \times Q$ matrix where each row $\{\mathcal Z_i\}_{i \le n}$ contains $Q-1$ zeros and a single one that represents membership at that given entry. We  introduce  indicator   $\Cb$  of length $Q$ whose values $C_q$ are 0 or 1  to determine if a block $q$ is signal or noise $NB$. 
	
	In addition to the latent variables and memberships, model  parameters $\Theta$ can be partitioned into $  \Theta_{\text{Signal}}$ and $\Theta_{\text{Noise}}$ in addition to global parameters $\boldsymbol{\alpha}, \Psi$:
	\begin{align} \label{eq:total_Theta}  
		\Theta = \{    \boldsymbol{\alpha}, \Psi  , \Theta_{\text{Noise}} , \Theta_{\text{Signal}} \}. 
	\end{align}   
	$\Theta_{\text{Signal}} = \{   \boldsymbol{\mu}_{q}    , \boldsymbol{\Sigma}_q\}_{q: 1 \le q  \le Q; B_q \ne NB}$ represents the model parameters that are unique to each block $B_q$ (not including $NB$),      and also there is one index $q_{NB}$ for noise block $NB$.
	$ \Theta_{\text{Noise}}= \{ \boldsymbol{ \mu}_{AN}, \boldsymbol{  \Sigma}_{AN}       \} $ represents the noise parameters that govern both interstitial noise $IN$ and noise block $NB$.   For $NB$, each correlation between $K$ layers is set at zero.

	The estimation procedure minimizes the hierarchical ELBO with respect to the parameters $\boldsymbol{\mu}, \boldsymbol \Sigma$ as well as memberships 
	The first term $ \ERhv \log f (\mathcal X,\mathcal Z)   $ which represents the observed joint densities of $\mathcal X$ and $\mathcal Z$ is written in Eq. \eqref{eq:ElogfXZ}.  $ \ERhv [ \log   R (\Cb ,\mathcal Z ) ] $ represents the joint distribution of the two-tiered variational variables and is written as: 
	\begin{align*} \Expec_{ R (\mathcal Z, \Cb  )}[\log R (\mathcal Z, \Cb  ) ] & = \sum_{i,q}    \tau_{iq} \log \tau_{iq} +   \sum_q    \bigg( P_q  \log P_q   +  (1-P_q)         \log  (1-P_q)    \bigg) .
	\end{align*}    
	The third term $\ERhv [ \log   S (\Cb| \mathcal Z  ) ] $  described by Ranganath et al.  as the `recursive variational approximation' \cite{ranganath2016}  for $R(\cdot )$, is 
	\begin{align*}
		\Expec_{ R (\Zb, \Cb  )}   \log  S (\Cb| \Zb  ) )   = 
		\sum_{i,q}  \bigg( P_q \log \Psi +(1- P_q) \log (1-\Psi)  \bigg)  \tau_{iq}.
	\end{align*}
	Combining the above terms, the hierarchical ELBO is written as:
	\begin{align*} 
		\mathcal{L}^\prime 
		=&  \ERhv [\log f (\mathcal X | \mathcal Z) ] +  \sum_{i,q}   \bigg(   \tau_{iq} \log \alpha_q    
		+ \tau_{iq} \log \tau_{iq}  +  \bigg( P_q \log \Psi +(1- P_q) \log (1-\Psi) \bigg)   \tau_{iq} \bigg) 
		\\
		&+  \sum_q    \bigg( P_q  \log P_q   +  (1-P_q)         \log  (1-P_q)    \bigg) .
	\end{align*}  
	\ch{In the following section we outline the EM framework and then discuss the derivations of $S(\Cb|\Zb)$ and $R(\Zb, \Cb)$.
	}
	Detailed derivations for all of these terms can be found in Appendix \ref{appendix:proofs_and_deriv}.

	The main innovation in our approach is in modeling joint approximate conditional distributions of $\Zb$ and $\Cb$ in addition to $\Zb$:
	\begin{align} \label{eq:R_ZC}
		R_{\mathcal X}(\Zb,\Cb  )  \approx  \prod_{i, q} \bigg(  m   (\Zb_i ,\boldsymbol{\tau}_i)  \times \text{Bern}(C_q, P_q)   \bigg) . 
	\end{align}  
	In Eq. \eqref{eq:R_ZC} $   R_{\mathcal X}(\mathcal Z,\Cb  )  $ represents the joint variational distribution of the memberships $\mathcal Z, \Cb$.
	The exact joint distribution is unknown, but the hierarchical mean field (MF) approximation  $   R(\mathcal Z,\Cb  )  $  can be used to obtain a factorized estimate for its marginals \cite{ranganath2016}.  We write the approximate composition of marginals using ``$\times$"; $m(\cdot)$ represents the multinomial distribution. The variational approximations of membership matrix $\mathcal Z$ is a $n \times Q$-dimensional matrix  $\boldsymbol{\tau}$, each row represents the vector of probabilities that approximates $\mathcal Z_i$ \cite{mariadassou2010}.  
	
	The variational approximation of the indicator $C_q$ at block $q$ is the probability $P_q$, which typologizes $\boldsymbol{\tau}$. Under variational distribution $R$, each member $i$ of a block $B_q$ adheres to multinomial distribution with parameter $\tau_{iq} = \Expec[Z_{iq}]$. $P_q$ is the probability of $C_q$ akin to $\tau_{iq}$. 
	$\Psi: = 1/Q$ is the \textit{prior} probability of block $\{B_q\}_{q: q\le Q}$ to be noise block $NB$.
	A derivation for $\Psi$ is given in Appendix \ref{app:derivPsi}.
	
	The hierarchical MF distribution $R_\text{hv}(\mathcal Z )$ as introduced by Ranganath et al. \cite{ranganath2016} ``marginalizes out" the MF parameters in $   R_{\mathcal X}(\mathcal Z ,  \Cb  )$ and is written as
	\begin{align*}
		R_\text{hv}(\mathcal Z ) = \int  R_{\mathcal X}(\mathcal Z ,  \Cb  )   \text{d}  \Cb.  
	\end{align*}
	Following the methods of estimation proposed in prior work on SBM estimation \cite{Daudin2008,mariadassou2010,paul_randomfx_2018}, 
	$R_{\mathcal X}(\mathcal Z, \boldsymbol{\tau} )$ represents the multinomial variational distribution wherein each  $\tau_{iq}$ approximates  the membership allocations.  $R_\text{hv}(\mathcal Z ) $ is the same as the variational distribution $R$ in prior work.
	\begin{figure}
		\centering
		\includegraphics[width=0.66\linewidth, trim= {0cm 0cm 0cm 0cm }]{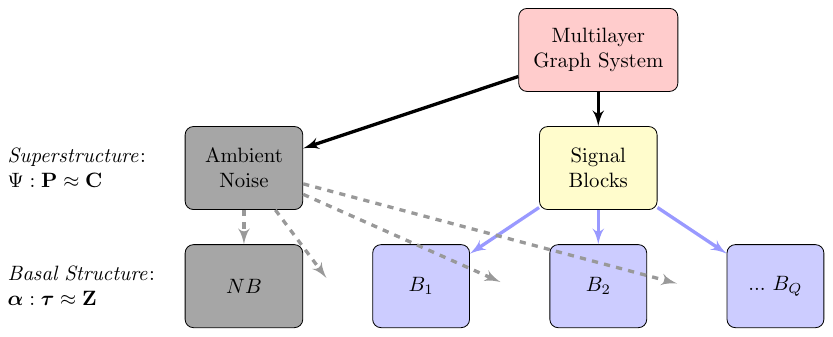}
		\caption{\footnotesize Schematic diagram for the hierarchy of organization for blockstructures with signal/noise differentiation for blocks as the top layer and the actual blocks as the bottom layer. }
		\label{fig:flowchart-signalnoise}
	\end{figure}
Prior VEM-based  estimation methods focus on optimizing ELBO \cite{paul_multilayer_2015,paul_randomfx_2018,mariadassou2010,Daudin2008}. 

$\mathcal{L} $ can be rewritten as the following 
\begin{align*} 
	\mathcal{L}   = \Expec_{ R_{\text{hv}} (\mathcal Z)  } [       \log \jointPxz]  + \mathcal{H}_{\hv}( R(\mathcal Z)),
\end{align*} 
where   $\mathcal{H}$  is the entropy of variational variable $\mathcal Z$ \cite{ranganath2016}.
A sharper bound than the ELBO is derived by introducing the marginal recursive variational approximation $S(\Cb|\mathcal Z)$, and then exploiting the following inequality with joint MF distribution $R(\mathcal Z, \Cb )$  and the (hierarchical) entropy $\mathcal{H} (\mathcal Z)$:
\begin{equation} \label{eq:HELBO_Ineq} 
	\mathcal{H}_\text{hv} ( R(\mathcal Z )) 
	\geq   - \Expec_{R ( \mathcal Z, \Cb ) }\left[    \log R (\mathcal Z, \Cb) \right]   + \Expec_{R (\mathcal Z, \Cb ) }  \left[  \log S( \Cb | \mathcal Z) \right]. 
\end{equation} 

The jointly  factorized mean field  components are $R ( \Cb  )$ and $ R (\mathcal Z| \Cb )$.  $R ( \Cb  )$ ia expressed as
$ R (\Cb )  = \prod_q    P_q    ^{C_q} (1-P_q)    ^{1-C_q}   $ and $R (\mathcal Z|\Cb  ) $ is written similarly to prior variational membership variables \cite{Daudin2008,mariadassou2010}, exponentiated by $C_q$:
$$R (\mathcal Z|\Cb  ) =    \prod_q  \bigg(   \prod_{i}  \tau_{iq}^{  Z_{iq} }  \bigg) ^{C_q}  \bigg(  \prod_{i} \tau_{iq}^{  Z_{iq} } \bigg)^{1-C_q},$$
combining to form   	$R (\mathcal Z, \Cb  ) =     R (\mathcal Z | \Cb ) R (\Cb ) .$
Moreover, the recursive variational  approximation  $S (\Cb| \mathcal Z  )$ \cite{ranganath2016} estimates the higher-order memberships $\Cb$   using the basal memberships $\mathcal Z$:
$$    S (\Cb| \mathcal Z  )   = \prod_q \prod_i \bigg(    \Psi ^ {C_q}    (1-\Psi) ^{ 1- C_q} \bigg) ^ {Z_{iq}}. $$ 


%

\section{Estimation Algorithm} \label{sec:est_algorithm}   
We summarize the targets of inference here to set up the language for the rest of the section. Variational parameters   $\boldsymbol{ \tau}_q$ and $P_q $  (for $q: q \le Q$) approximate the membership allocations, while model parameters describe the parametric qualities of the blocks.   Within the set of model parameters, we further distinguish  \textit{local} and \textit{global} parameters. \textit{Local} block-wise parameters are represented by $\boldsymbol{\Theta}_q$, and membership probabilities $\alpha_q$ for each $q$. \textit{Global} parameters are $ \Psi  , \Theta_{\text{Noise}} $. 
We use  VEM to estimate variational parameters in the E-step and model parameters in the M-step, alternating these steps until the differences in $\btau$ become miniscule.  We present the closed-from solutions to all the estimates below, but detailed derivations for every term is found in Appendix \ref{app:VEM_Calc}. Operationally, the E-step and M-step are implemented in \ch{an} alternating fashion until the membership variables $\btau$ converge. 

First we introduce some more terms 
\begin{align} 
	f( G^k_{ij}, \boldsymbol{\mu}_q, \boldsymbol{\Sigma}_q )  &= \frac 1 2   ( G^k_{ij}     - \boldsymbol{\mu}_{q}   )^T  \boldsymbol{\Sigma}_q^{-1} (G^k_{ij}      - \boldsymbol{\mu}_{q}   )         -   (2\pi)^{K/2}  (\log    |\boldsymbol{ \Sigma}_q |  )^{  1 /2 }
	\stepcounter{equation}\tag{\theequation} \label{eq:abbrev_Sigfreq}
	\\
	f( G^k_{ij}, \boldsymbol{\mu}_{AN}, \boldsymbol{\Sigma}_{AN} ) &= \frac 1 2   ( G^k_{ij}     - \boldsymbol{\mu}_{AN}   )^T  \boldsymbol{\Sigma}_{AN}^{-1} (G^k_{ij}      - \boldsymbol{\mu}_{AN}   )         -   (2\pi)^{K/2}  (\log    |\boldsymbol{ \Sigma}_{AN} |  )^{  1 /2 }
	\stepcounter{equation}\tag{\theequation} \label{eq:abbrev_Noifreq}.
\end{align} 
Equation \eqref{eq:abbrev_Sigfreq} denotes the density for  edges in a signal block $\big( \boldsymbol{\mu}_q , \boldsymbol{\Sigma}_q \big)$ at layer $k$;  equation \eqref{eq:abbrev_Sigfreq} denotes density for edges with noise $\big( \boldsymbol{\mu}_{AN} , \boldsymbol{\Sigma}_{AN} \big)$.
\ch{Graph $\Gb$ with $K$ graph-layers $\{\Gb^1,..., \Gb^K\} $, has conditional density 
	\begin{align*}
		\log f   (  \Gb| \mathcal Z  )  
		= &
		\sum_{q : B_q \ne NB;q \le Q } \sum_{k \le K }   
		\sum_{i , j \le n }   \    Z_{iq} Z_{jq}    f( G^k_{ij}, \boldsymbol{\mu}_q, \boldsymbol{\Sigma}_q) +   \\
		\textbf{1} (B_q& = NB )   \sum_{i , j \le n}     Z_{iq} Z_{jq}  f( G^k_{ij}, \boldsymbol{\mu}_{AN}, \boldsymbol{\Sigma}_{AN} ) 
		+   \sum_{q,l \le Q : q \ne l }       \sum_{i,j  \le n  }   Z_{iq} Z_{jl}  f( G^k_{ij}, \boldsymbol{\mu}_{AN}, \boldsymbol{\Sigma}_{AN} ) .
		\stepcounter{equation}\tag{\theequation} \label{eq:fXZ}
	\end{align*}
	The log likelihood portion of the ELBO, $\log( f   (  \Gb| \mathcal Z  )  )$, written above in Equation \eqref{eq:fXZ} is comprised of three parts: unique signals for every $q$ (top), the noise block $NB$ (bottom left), and the interstitial noise $IN$ (bottom right).
	$AN$ is the global \textit{ambient noise}  whose parameters govern the \textit{interstitial noise} as well as \textit{noise block} as in Definition \ref{def:ambient_noise}.
	Given variational variables   $\boldsymbol{\tau}, \Pb$, the expected likelihood is 
	\begin{align*}
		\Expec_{R_{\bG}} [\log( f   (  \Gb| \mathcal Z  )  ) ] 
		=& 
		\sum_{q : q \le Q}     \Prob(B_q  \ne NB)  \tau_{iq}  \tau_{jl}    f( G^k_{ij}, \boldsymbol{\mu}_q, \boldsymbol{\Sigma}_q)   
		\\
		+  \Prob(B_q =  NB) &     \sum_{i , j: i \ne j }         \tau_{iq} \tau_{jl}  f( G^k_{ij}, \boldsymbol{\mu}_{AN}, \boldsymbol{\Sigma}_{AN} ) 
		+      \sum_{q, l \le Q: q \ne l   }       \sum_{i , j \le n }         \tau_{iq} \tau_{jl}  f( G^k_{ij}, \boldsymbol{\mu}_{AN}, \boldsymbol{\Sigma}_{AN} ) .
	\end{align*} 
	The $   \ERhv [\log f (\mathcal Z ) ] $ term restores to the same form as earlier work on SBMs \cite{mariadassou2010,Daudin2008}:
	\begin{equation}  \label{eq:fZb_same}
		\ERhv [\log f (\mathcal Z ) ] 
		=    \sum_{i,q}       \tau_{iq} \log \alpha_q,    
	\end{equation}     
	where as in prior work \cite{Daudin2008,mariadassou2010}, the variables $\alpha_q$ represent the  membership probabilities of  $Z_{iq}$ and sum to 1:
	\begin{align}\label{eq:alpha_q}
		\alpha_q = \Prob(i \in B_q) = \Prob( Z_{iq}=1).
	\end{align}
	For the rest of the manuscript we use  $  \sum_{i,q}     ( \cdot )$ to signify the double summation across all $i \le n$ and $q \le Q$.  The expected log frequency of the membership vectors $\mathcal Z$ reduces to that in canonical SBMs. Details of this identity are found in Appendix \ref{app:preserve_ElogfZ}.} The joint density is written as:
\begin{equation} \label{eq:ElogfXZ} 
	\Expec_{ R (\mathcal Z, \Cb  )}[\log f (\Gb , \mathcal Z )]  =     
	\Expec_{ R_{\bG}}
	[\log f (\Gb | \mathcal Z ) ] + \sum_{i,q}     \tau_{iq} \log \alpha_q .  
\end{equation} 
The expression is written in full in Appendix \ref{app:ElogfX|Z}.

\subsection{E-Step} \label{sec:E_stepTauEst}
The E-Step of the algorithm estimates the variational variables which represent block memberships $Z_{iq}$ of the nodes $i$ as well as  $C_q$  which  represent the  ``memberships of memberships". First we describe the estimation procedure for the variational approximations $\tau_{iq}$, next we describe the estimation of signal-noise differentiation probabilities  $P_q$. This two-step procedure differs from prior work because of an additional hierarchical estimation step of the higher-level variational variables $P_q$.

$\tau_{iq} $ is estimated by an iterative fixed-point approach.   Derivatives for each $\tau_{iq}$ are calculated based on model parameters and $\tau_{jl}$, 
\begin{align*} 
	\log(\tau_{iq})   \propto   \log(\alpha_q )
	&+     \sum_{k \le K  }  \sum_{j  \le n }   \tau_{jl} \bigg(        P_q f(G^k_{ij}, \boldsymbol{\mu}_q, \boldsymbol{\Sigma}_q )   + (1-P_q)  f(X^k_{ij}, \boldsymbol{\mu}_{AN}, \boldsymbol{\Sigma}_{AN} ) \\
	&+    \sum_{l \le Q ; l \ne q }  f(G^k_{ij}, \boldsymbol{\mu}_{AN}, \boldsymbol{\Sigma}_{AN} )    \bigg)     -1  + P_q \log \Psi +(1- P_q) \log (1-\Psi)  .
\end{align*}
After exponentiating, the fixed-point equation can feasibly be solved after the iterating the system until relative stability.  This is the same approach as most existing literature \cite{Daudin2008,mariadassou2010}.  $P_q$ are calculated as follows:
\begin{align} \label{eq:Pq_calculation}
	\widehat{P_q} =   1 - \bigg(   1+   \bigg[  \exp  \bigg(   \sum_{k \le K }  \sum_{i,j \le n}   \tau_{iq} \tau_{jq}  \bigg(     f(G^k_{ij}, \boldsymbol{\mu}_q, \boldsymbol{\Sigma}_q ) - f(G^k_{ij}, \boldsymbol{\mu}_{AN}, \boldsymbol{\Sigma}_{AN} )  )   +   \log  \bigg(
	\frac{1-\Psi }{ \Psi }
	\bigg)   \bigg)  \bigg) \bigg]  ^{-1}   \bigg) ^{-1}. 
\end{align} 
Calculations for each of these terms are provided in Appendices \ref{app:tau_est} and \ref{app:noise_probPq}.
We  apply  stochastic variational inference (SVI) to calculate the membership parameters $\tau_{iq}$ and $P_q$. Details for SVI are described in Appendix \ref{app:StochVI}.

\subsection{ M-Step} \label{sec:Mstep}
Similar to its estimation in Daudin et al.  \cite{Daudin2008}, $ {\alpha}_q $ are estimated as follows using Lagrangian multipliers: $   \hat{\alpha}_q = {  \sum_{i,q} \tau_{iq} } / { n}.  $
The closed-form estimate for the \textit{local} parameters for the mean vector $   {\boldsymbol{ \mu } }_{q} $ for each block $q$   from the M-step is
\begin{align*}
	\widehat{\boldsymbol{ \mu } }_{q} & = \frac{ \sum_{i,j} \tau_{iq} \tau_{jq} \Gb_{ij}     }{ \sum_{i,j} \tau_{iq} \tau_{jq}    }   P_q +     \boldsymbol{ \mu}_{AN}    (1- P_q ).
\end{align*}

In the above, and all subsequent expressions in this subsection, the derivations are located in Appendix \ref{app:M_StepSignalTerms}.
Similarly to mean calculations, the variance calculations (along diagonals) are 
\begin{align*}
	\widehat {  \boldsymbol{\Sigma}_q}  &  =   \frac{  \sum_{i,j}   \tau_{iq}  \tau_{jq}  ( \Gb_{ij} - \boldsymbol \mu_q  )^2     } {  \sum_{i,j}   \tau_{iq}  \tau_{jq}    }   P_q   
	+  \boldsymbol \Sigma_{{AN}}  (1-P_q ). 
\end{align*}  
The  cross-term for two layers  $h, k$  is written as:
\begin{align*} 
	\widehat{\boldsymbol  {\Sigma}}_{hk,q}
	&=\frac{   \sum_{i,j}   \tau_{iq}  \tau_{jq}  (G^k_{ij} -  \mu_{q,k}) (G^h_{ij} -   \mu_{q,h}  )         } {  \sum_{i,j}   \tau_{iq}  \tau_{jq}   }  P_q.
\end{align*}
The element-wise  correlations at iteration $t$ across layers  $h,k$ ($h \ne k$)  are then calculated,   and the maximum (if $K>2$)  of these values is taken as the putative correlation (across all layers) for block $q$  
\begin{align*}
	\hat{\rho_q}  =   \max_{h,k}      \frac{\widehat{ {\Sigma_{hk}^q}} }{    
		\sqrt{\widehat{ {\Sigma^{ h}_q}} \widehat{ {\Sigma^{ k}_q}}}}. 
\end{align*}
If $K=2$ then no maximum needs to be taken. This is an operational step of the optimization and does not necessarily yield closed-form estimates.   This value is also known as the \textit{mutual coherence} of estimated correlation matrix and serves as a summary statistic of the estimates for correlations \cite{Tropp_JustRelax2006}.   

\subsubsection{Estimation of Global Parameters } 
\newcommand{\ttau}{\tilde{\tau}}
\newcommand{\nottau}{\dot{\omega}}

At each iteration of VEM, the closed-form solutions of the global parameters $	{ \widehat{   \boldsymbol{\mu}} _{AN}    }$ and 
$		{ \widehat{\boldsymbol{\Sigma}}_{AN}}$
are written as follows.  $	{ \widehat{   \boldsymbol{\mu}} _{AN}    }$ is 
\begin{align}
	{ \widehat{   \boldsymbol{\mu}} _{AN}    }
	& =    \Psi  \frac {  \sum_{j , i  }    \sum_{l, q: q \ne l } \tau_{iq}  \tau_{jl}   
		\Gb_{ij}  
	} {  \sum_{j , i  }    \sum_{l, q: q \ne l }     \tau_{iq}  \tau_{jl}         } 
	+(1-\Psi)
	\frac{	  \sum_{j , i  }    \sum_{ q }    \tau_{iq}    \tau_{jq}   (1-P_q)
		\Gb_{ij}  }{   \sum_{j , i  }    \sum_{ q }    \tau_{iq}    \tau_{jq}     (1-P_q) }.
\end{align} \label{eq:mu_AN} 
The variance of global parameters is similarly calculated as:  
\begin{align*}
	{ \widehat{\boldsymbol{\Sigma}}_{AN}}
	& = \Psi \frac {  \sum_{j , i  }    \sum_{l, q: q \ne l } \tau_{iq}  \tau_{jl}   
		(\Gb_{ij} - \boldsymbol{\mu}_{AN} )^2  
	} {  \sum_{j , i  }    \sum_{l, q: q \ne l } \tau_{iq}  \tau_{jl}    } 
	+    (1 - \Psi) \frac {   \sum_{j , i  }    \sum_{ q }  \tau_{iq}  \tau_{jl}  (1-P_q) 
		(\Gb_{ij} -   \boldsymbol{\mu}_{AN}   )^2      } {    \sum_{j , i  }    \sum_{ q }   \tau_{iq}    \tau_{jq}   (1-P_q)   } ,
\end{align*} 
the covariance term for global noise, as stated earlier, is zero. 
Derivations for these expressions are in Appendix \ref{app:deriv_muAN}.

\section{Case Study: PNC Psychopathology Networks} \label{sec:results_PNC}
We apply \texttt{SBANM} to the PNC data as the primary case study of this paper.  We first construct networks from \textit{anxiety, behavior}, and \textit{mood} psychopathologies as described in Section \ref{sec:data}, then run the algorithm and  subsequently cross validate and empirically verify the results with diagnoses data. 
We use the notation for data outlined in Section \ref{sec:data}: $\Xb$ represents the layer of symptom response networks for anxiety, $\Yb$ for behavior, $\Zb$ for mood disorders. Correspondingly, we let  $ \big( \boldsymbol \mu_{\Xb}, \boldsymbol \mu_{\Yb}, \boldsymbol \mu_{\Zb} \big)$ represent the means of the edge-connections for each block representing anxiety, behavior, and mood with corresponding standard deviations $  \big( \boldsymbol \sigma_{\Xb} , \boldsymbol \sigma_{\Yb}, \boldsymbol \sigma_{\Zb} \big)$. 

Not much prior work has  approached the study of psychiatric conditions using subject-networks. We construct networks of individuals as nodes and their similarity as edges. Distinct conditions are represented by different layers as in a multilayer network.  The goal of introducing \textit{ambient noise} to psychopathology symptom networks is to  identify groups of people who have similar clinical characteristics  and facilitating early identification of individuals at high risk of developing the disorder, in this case psychosis spectrum. 
Existing machine learning studies of psychosis spectrum  typically require input from  already-diagnosed subjects. These analyses usually use methods such as logistic regression \cite{dcannon_psyRiskCalculator2016}.   However, we aim to classify anxiety, mood, and behavior symptoms to identify who is at risk for psychosis \textit{without} the knowledge of which patients have psychosis spectrum.

Unsupervised analysis is useful in early identification in clinical settings, and we leverage the \texttt{SBANM} method to conduct exploratory analysis that will pave way for potential evidence-based  intervention schemes.
The developmental periods prior to the onset of psychotic disorders are critical targets of early intervention and as such serve appropriate data for experimental hypotheses of `exploratory clustering and classification for the purpose of early detection.


\ch{\subsection{Scientific Hypothesis}}\label{sec:scientific_hypothesis}
We have introduced literature in Section \ref{sec:motivation} that details specific timing for onset of psychosis in early adult subjects \cite{Cupo2021,Krabbendam2005DevelopmentOD,Tandon2012AttenuatedPA}. 
Tandon et al. and Cupo et al.  posits a qualitative change in subjects' psychopathologies as they transition from the prodromal stage into the psychosis stage  \cite{Tandon2012AttenuatedPA,Cupo2021}. 
Existing research on pre-psychotic psychopathologies  note that ``psychotic disorders may be due to \textit{nonpsychotic} common mental disorders such as depression and anxiety" \cite{Cupo2021}.
\ch{Cupo et al. document that ``epidemiological  cohorts also demonstrate co-occurring and reciprocal relationships" between these disorders and psychosis. A myriad of interacting psychopathologies, notably anxiety, behavior disorders, depression, mood disorders known as the psychosis \textit{prodrome} are demonstrated to {precede} the \textit{first episode} of psychosis \cite{Cupo2021,Krabbendam2005DevelopmentOD,Chen2019PatternsOS}. After the first episode, however, the diseases progresses out of the \textit{prodrome} and into ``full-blown psychotic illness":} \toch{several works have described this decoupling, but few have statistically modeled such a transition \cite{Tandon2012AttenuatedPA,Renwick2015ProdromalSA}.}

We seek to separate the subjects that have transitioned  to psychosis from those who did not. We model the \textit{co-occurrence} among common prodromal symptoms as the correlations between the multilayer network $\Gb$ constructed from anxiety ($\Xb$), behavior  ($\Yb$), and mood  ($\Zb$) disorders (Section \ref{sec:data}).  
Prior work suggest that  there is separation among \textit{independent} and \textit{correlated} groups of subjects (in $\Gb$) \cite{Tandon2012AttenuatedPA}.
\toch{We hypothesize that since prodromal symptoms are highly correlated \cite{Chen2019PatternsOS,Cupo2021}, and that they are not associated with non-initial symptoms of psychosis \cite{Renwick2015ProdromalSA}},
the sample that has converted to psychosis from the prodrome \cite{Tandon2012AttenuatedPA}
will have \textit{independent}, but exhibit\textit{ high rates of}, prodromal symptoms.  

We have also traced the literature on sex differences among such co-occurrences between common psychopathologies and their relationship with the first episode of psychosis \cite{Mc1988GenderDI,Nowrouzi2015AgeAO,Kirkbride_2012_Sexdiff,Li_SexDiff_2016}.
We restricted the analysis to the early adult sample for this study, and split up the sexes among subjects to examine the potentially differential effects of clustering (Section \ref{sec:data}).
Li et al. \cite{Li_SexDiff_2016} cite several other   works in describing the difference in the peaks of rates of psychosis onset between sex \cite{Mc1988GenderDI,Nowrouzi2015AgeAO,Kirkbride_2012_Sexdiff}. The consensus among literature describe the peaks of onset as between 21-25 for males, and 25-30 with another peak occurring much later in the middle ages for females. Indeed, for the PNC sample to overlap with the range of psychosis onset, the target sample is male early adults aged 18-21. The sample size is thus 764 subjects.

\ch{\subsection{Clinical Verification}} \label{sec:results_clinical}

We ran the \texttt{SBANM} algorithm on early adult PNC subjects stratified by sex. We set $Q=3$ based on the optimal \text{Integrated Composite Likelihood} \cite{matias2016}, of which a more detailed explanation is provided in Section \ref{sec:expeirmental_procedure}. 
Table  \ref{fig:psychtable} shows the average proportions of subjects who met the criteria for clinical diagnoses of anxiety, mood, and behavior disorders, psychosis spectrum as well as those who were typically developing (TD).  The columns after block labels and sizes are positive indicators for anxiety, behavior, and mood disorders.  Each clinically identified indicator is  `yes' or `no' for each subject.
Among males, the results remarkably differentiate rates of psychosis between the $NB$ group and the other correlated clusters (Table \ref{fig:psychtable}, left). However, similar rates of differentiation are not found among females  (Table \ref{fig:psychtable}, right). 
Furthermore, the rates of psychosis in the independent block $NB$ from clinical verification is 54\%, while none fall under \textit{typical development} (TD) among males.

The high rates of psychosis spectrum among males in $NB$ coincide with the clusters where anxiety, mood, and behavioral disorders are disjoint support the hypothesis that the \textit{first episode of psychosis} marks a qualitative transition from prodrome to psychosis spectrum. Furthermore, the timing (in early adult) and difference in the distinguishing characteristics among $NB$ blocks between sex also concur with the prior work in timing of psychosis onset.  The most significant clustering result  is found among subjects in the $NB$ block  (Table \ref{fig:psychtable}). Among these subjects, their high rates of psychosis,  and low (0\%) rates of \textit{TD}  lends evidence of a psychosis spectrum conversion group \cite{Tandon2012AttenuatedPA}. The high rates of anxiety, behavior, and mood disorders persist in spite of their independence among layers  indicate that these prodromal signs persist, but become decoupled \cite{Cupo2021,Renwick2015ProdromalSA}. 
Uncorrelated symptoms among these subjects in $NB$ could suggest that they tend towards psychosis through individuated channels.

\begin{table}[htbp]
	\ch{    \textsc{Psychopathological Symptom Groupings (Early Adult (18-21))} }
	\par \smallskip
	\scriptsize
	\centering 
	\begin{tabular}{ccc|ccccc|}
		\hline
		\hline 
		\multicolumn{8}{c}{\textsc{ Male}}  \\
		\hline 
		&Block	& $n$ & \textbf{Anx} & \textbf{Beh} & \textbf{Mood} & TD & \textbf{Psy} \\ 
		\hline
		\Large $\bullet$ &	$NB$ & 41 & 73 & 95 & 46 & 0 & 54 \\ 
		\color{Yellow}\Large $\bullet$	& $S1$ & 244 & 51 & 39 & 29 & 30 & 38 \\ 
		\color{SkyBlue}\Large $\bullet$		& $S2$ & 471 & 39 & 21 & 9 & 42 & 14 \\ 
		\hline
	\end{tabular}
	\begin{tabular}{ccc|ccccc}
		\hline
		\hline 
		\multicolumn{8}{c}{\textsc{Female}}  \\
		\hline 
		&Block	& $n$ & \textbf{Anx} & \textbf{Beh} & \textbf{Mood} & TD & \textbf{Psy} \\ 
		\hline
		&	$NB$ &  35 & 11 & 17 & 3 & 69 & 20 \\ 
		& $S1$ &883 & 65 & 25 & 26 & 24 & 17 \\ 
		& $S2$ & 189 & 36 & 16 & 12 & 53 & 25 \\ 
		\hline
	\end{tabular}
	\caption {\footnotesize  \textit{Mean summary statistics for psychiatric diagnoses (approximate diagnostic criteria of DSM-IV) for early adult males and females. The following columns details symptoms of anxiety, behavior, and mood disorders. The `Psy' column gives the average of whether the respondants have overall diagnoses for psychosis and the `TD' column indicates typical development.}}
	\label{fig:psychtable}
\end{table}


Psychosis rates are clearly differentiated between different blocks; those in $NB$ are consistently higher. The differential clustering results for early adult males likely demonstrate latent neurodevelopmental pathways for onset of psychosis.  Psychosis onset is characterized by presence of active psychotic symptoms and occurs during early adulthood between 21-25 for males \cite{Li_SexDiff_2016}. This represents a continuum with individuals reporting proportionally more depression, anxiety, and behavior psychopathology  prior to the onset of psychosis \cite{Cupo2021}.  As symptoms segregate with growth and become statistically independent, clustered subjects with higher correlations $\rho_q$  correspond to  more interconnected  pre-psychotic pathways \cite{Chen2019PatternsOS}, while subjects with independent symptoms are indicative of progressing past the first episode of psychosis \cite{Renwick2015ProdromalSA,Tandon2012AttenuatedPA}. That these categories emerged without any supervision demonstrates the discerning ability of \texttt{SBANM}. Results did not show any strong differentiation in other demographic characteristics (Table \ref{fig:ClusDemog_table} in Appendix \ref{app:pnc_posthoc}). 

\ch{\subsection{Method Comparison for PNC Data}}\label{sec:method_comparison}

We compare different community detection methods to cluster the PNC data.
In the absence of ground-truth data for clusters among real data, we consider the diagnoses data of PNC subjects to validate results. 
The zero-correlation constraint between the layers $\Xb, \Yb, \Zb$ within $NB$ discovered by \texttt{SBANM} is natural for testing our clinical hypothesis (Section \ref{sec:scientific_hypothesis}).
\toch{We interpret $NB$ as the group of subjects that have transitioned from the prodromal stage to psychosis spectrum.}
We compare the clusters with the highest psychosis rates that was obtained from each method; we define  $q^*$ as the cluster that yields the partition of subjects with the highest rates of psychosis. Table \ref{table:method_compare} (right) compares the characteristics of  $q^*$ for each method by taking their average rates of anxiety, behavior, and mood disorders, in addition to those of psychosis and typical development (TD).  $Q$ for each method is assessed based on their own internal criteria for best fit. 
\texttt{dynsbm} (row 2) for example finds 8 groups  based on its own ICL criteria.  Optimal $Q$ was between 3 or 4
for most spectral methods.  
The results of \texttt{SBANM} (Section \ref{sec:results_clinical}), yielded 41 subjects in $NB$ ($q^*$) with average psychosis rates of 54\%. The proportion of subjects that approximate criteria for a clinical anxiety disorder from post-hoc evaluations is 73\%, behavioral disorders  was 95\%, and 46\% for mood disorders.
\toch{Out of all the  methods,   \texttt{SBANM} and \texttt{dynsbm} found the clusters with the highest rates of psychosis.}
\texttt{dynsbm} finds a very small group (of 14 subjects) that has high rates of psychosis, anxiety, behavior, and mood disorders.

The identified clusters from \texttt{dynsbm} and \textit{MASE}, both yield  strongly positively correlated blocks across layers. Interpretating these 
results as symptoms transitioning from the \textit{prodrome} to \textit{psychosis spectrum} is less suitable.  
\toch{The high correlations signify high rates of co-occurrence and  may correspond to the correlated prodromal symptoms that precede psychosis \cite{Tandon2012AttenuatedPA}.}\ch{ However, they do not capture the qualitative change 	that we posit as the transition from psychosis prodrome to psychosis spectrum as outlined in Section \ref{sec:scientific_hypothesis} \cite{Renwick2015ProdromalSA,Chen2019PatternsOS,Cupo2021}}.
Other methods  yield results with similar rates of psychosis ($\approx$50\%) , TD($\approx$0\%)   anxiety ($\approx$75\%) and mood disorders ($\approx$50\%). The spectral methods (rows 2,4,5,6) typically beget much larger groups of around one quarter to one third of the total sample size. 
These larger, evenly populated subgroups may yield some advantages, but reveal less specificity in terms of potential diagnosis.  
The constraint of independence (through zero correlation) allows a much more specific demarcation of varying psychopathologies. 

We also compare multivariate spectral clustering (\texttt{MVSPEC}) to the PNC data that was not transformed to networks (row 6). This method did not separate subjects with psychosis nearly as well as the network methods. The degradation in classification  suggests that the network transformation fo large-scale questionnaire (or survey) data is perhaps  even necessary for clustering analysis with the goals of diagnosis and prevention. This method was not evaluated for the simulated data.



\begin{table}[htbp] \footnotesize
	\centering
	
	\begin{tabular}{cc}
		\hline
		\hline
		\multicolumn{2}{c}{\textsc{ Method Comparison}}  
		\\ 
		\hline
		\begin{tabular}{lcc}
			\multicolumn{1}{c}{}  &	\multicolumn{2}{c}{\textbf{Simulations (50 Runs)}}  \\ 
			\hline
			\hline
			\textit{Method}& \textit{NMI} &\textit{ARI}  \\ 
			\hline
			&  Mean $\pm$ SD &  Mean  $\pm$ SD \\ 
			\hline
			\texttt{SBANM} & 1   & 1  \\ 
			\url{SPEC}(sum) &   .54$\pm$.02  & .53$\pm$.01 \\ 
			\url{dynsbm} & .98$\pm$.04 & .95$\pm$.02\\ 
			\url{MASE}&  .90 $\pm$.05 &  .76$\pm$.09 \\ 
			\url{MVSPEC}(net) &  .66$\pm$.02  &  .55$\pm$.14   \\ 
			\url{MVSPEC}(raw) & \textit{NA}   & \textit{NA}    \\ 
		\end{tabular}
		&  
		\begin{tabular}{ccccc|ccc}
			\multicolumn{8}{c}{\textbf{PNC Data (Male EA)}}  \\ 
			\hline
			\multicolumn{8}{c}{Characteristics (\%) of $q^*$ ($q$ with highest \%\textbf{Psy})} 		\\
			\hline
			\textbf{Psy}&\textbf{TD}& Anx&Beh&Mood&   $n_{q^*}$ &  $\boldsymbol{\rho}_{q^*}$&$Q$ \\ 
			\hline
			54&  0& 73 &95&46& 41 & 0,0,0 &  3\\ 
			51& 2 & 77 &77 &49 & 171& -.1,-.2,.4   &3  \\ 
			57& 0 & 71 &93&57 &14&  .3,.6.,7 & 8 \\ 
			52& 5 &76 & 55& 50 &  139 & .5,.5,.3  &4  \\ 
			42& 4  & 74 &72 & 52 &      264  &  0,.1,.3 & 4 \\ 
			33 & 20 & 62 &56&42& 194&   $NA$ &4 \\ 
		\end{tabular}
		\\
		\hline
	\end{tabular}

	\caption{\footnotesize \textit{Comparison of different methods for membership recovery using the ARI and NMI measures.   \texttt{dynsbm}  (unique config.) refers to the interpretation of the method when every unique configuration of blocks across layers are treated as a unique block. \texttt{dynsbm}  (most freq.)  treats the block with the most frequent occurence of memberships across all layers as the cross-layer block.
			In the right column, the
			$q^*$ represents cluster with highest \%\textbf{Psy}.
		}
	} \label{table:method_compare}
\end{table}


\section{Synthetic Experiments}\label{sec:app_results} 
In this section we describe the simulation studies to demonstrate the accuracy and efficacy of the proposed method.  \ch{Simulations are generated to match the outcomes of the real data in the previous section.} We considered networks of  three layers with size 800 that match approximately with the results of early adult males.
\ch{We simulate 50  networks with underlying memberships and parameters that approximately match those in Section \ref{sec:results_clinical}}. We then run \texttt{SBANM} on these networks to demonstrate that the method is able to recover simulated memberships and parameters. 
We also assess the computation times of various simulations and compare them to existing methods. 
The estimation algorithm is more parsimonious and highlights more nuanced relationships compared to some existing methods described in the following Section \ref{sec:method_comparison}.

\subsection{Experimental  Procedure} \label{sec:expeirmental_procedure}

The goal of these experiments is to demonstrate that our proposed method can faithfully recover generated memberships and parameters. 
For this section, we write 
$ \boldsymbol \mu := \big( \boldsymbol \mu_\Xb, \boldsymbol \mu_\Yb, \boldsymbol \mu_\Zb \big)$ represent the means of the edge-connections for each block representing anxiety, behavior, and mood with corresponding standard deviations $ \boldsymbol \sigma =  \big( \boldsymbol \sigma_\Xb, \boldsymbol \sigma_\Yb, \boldsymbol \sigma_\Zb  \big)$.  We let $\boldsymbol \rho$ represent the block-wise correlations  (across all layers) $\boldsymbol \rho:= (\rho_1,...,\rho_Q)$. 
Blockwise parameters  $\boldsymbol{\mu}, \boldsymbol{\sigma},$ and $\boldsymbol \rho$ are extracted from the early adult males results from Section \ref{sec:results_clinical} to serve as the ground-truth parameters for the following experiment.
Probabilities of membership-allocations $\boldsymbol \tau$ are also extracted from the data results and used to generate multinomial distributions that approximate the ``true" distributions of memberships. As such, simulated block-sizes are randomized but approximately match those of the data.

For every network, nodes are simulated within clusters with membership probabilities $\boldsymbol{\tau}$. Within these clusters,  edge-weights are simulated according to  the multivariate normal parameters   $\boldsymbol{\mu}, \boldsymbol{\sigma},$ and $\boldsymbol \rho$. 
Simulated data is generated after extracting these ground-truth parameters from the results of the PNC early adult males.
We set  $n$ to be 800, and then simulate a multinomial distribution of fixed total size $n$ where each cluster has membership probabilities of 4\%($NB$), 32\%($S_1$), and  62\%($S_2$) from  Table \ref{fig:psychtable}.
For each mean-covariance pair corresponding to block $q$, we generate multivariate Gaussian distributions with a sample size of $n_q (n_q -1)/2$, then  we convert these multivariate data to weighted edges. Finally,  a  sample of the $AN$ distribution  with size
\ch{$$n_{IN} := (n-1)n/2 - \sum_{q = 1}^Q n_q (n_q -1)/2$$} is generated for all  $ n_{IN}$ interstitial edges  between differing blocks.  \texttt{SBANM} is then applied to these networks and we assessed membership as well as parameter recovery .


Fifty three-layer networks were generated from a fixed set of parameters and membership probabilities.
For the \texttt{SBANM} algorithm, the initial membership probabilities  $\boldsymbol{\tau}$ are obtained by  by applying spectral clustering on the sum graph $\tilde{X}_{ij} = X_{ij} + Y_{ij} + Z_{ij}.$  across all $K$ layers, then averaged with uniformly generated probabilities. 
Results show consistently accurate estimates for the mean, variance, and correlation parameters (Figure \ref{fig:simu_boxplots}).
The algorithm was able to  exactly recover memberships for all simulations (\cite{abbe2017community}).   Table \ref{table:method_compare} shows that  \texttt{SBANM} is able to retrieve the  simulated memberships at  a 100\%  recovery rate.   True parameters are shown in Figure \ref{fig:simu_boxplots}.The variances for most of the estimates were within 0-3\% of the true values. More simulations are described in Appendix \ref{appsec:simu_experiment2}.

\begin{figure}[htbp] 
	\centering
	\begin{tabular}{ccc}
		\hline
		\multicolumn{3}{c}{\textsc{  Boxplots of Estimated Parameters}}  \\ 
		\hline
		&	$\boldsymbol{\mu}_{k,q}$&	$\boldsymbol{\Sigma}_{k,qq}$ \\
		\rotatebox[origin=l]{90}
		{ 	\quad \quad 	\quad \quad \quad \quad	\quad  \footnotesize    $\Xb$} &
		\includegraphics[width=0.4\linewidth]{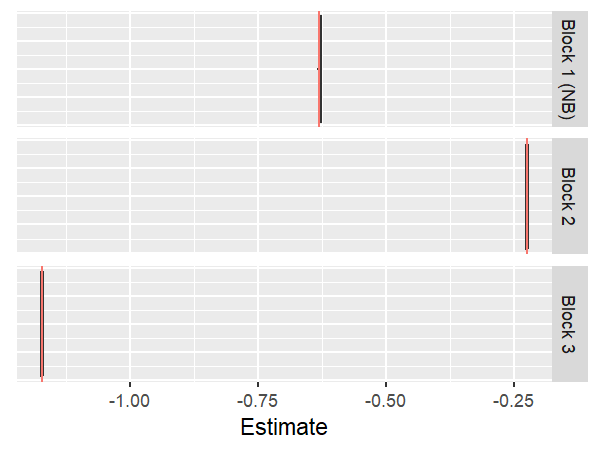}	& 	\includegraphics[width=0.4\linewidth]{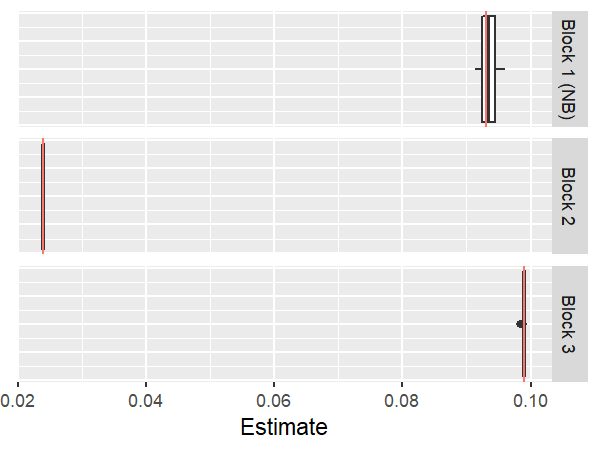} 	 	\\  
		\rotatebox[origin=l]{90}{ 	\quad \quad \quad  \quad 	\quad \quad \quad	  \footnotesize    $\Yb$}
		& \includegraphics[width=0.4\linewidth]{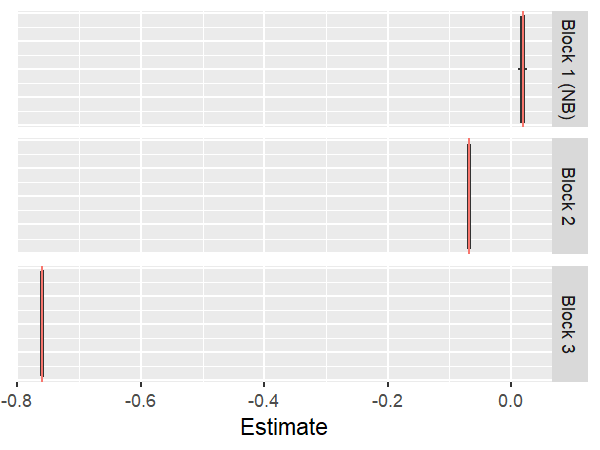}	&    \includegraphics[width=0.4\linewidth]{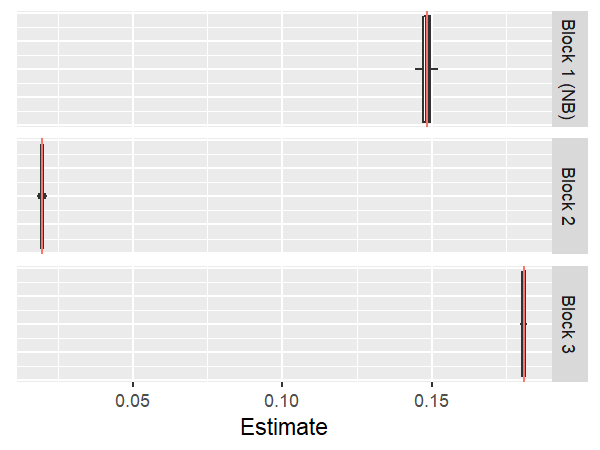} \\
		\rotatebox[origin=l]{90}{ 	\quad \quad \quad \quad  	\quad \quad 	\quad  \footnotesize    $\Zb$} 
		&\includegraphics[width=0.4\linewidth]{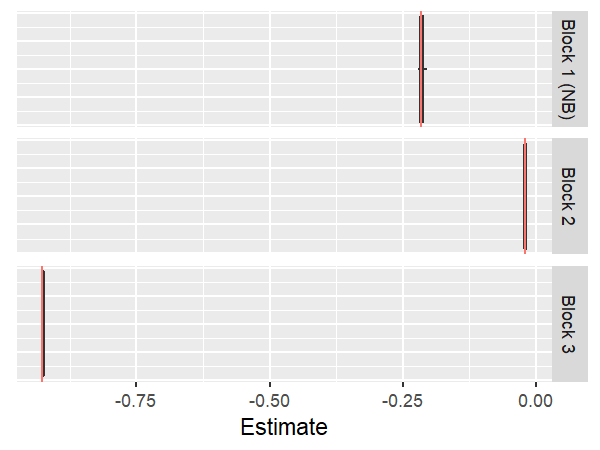}  
		&\includegraphics[width=0.4\linewidth]{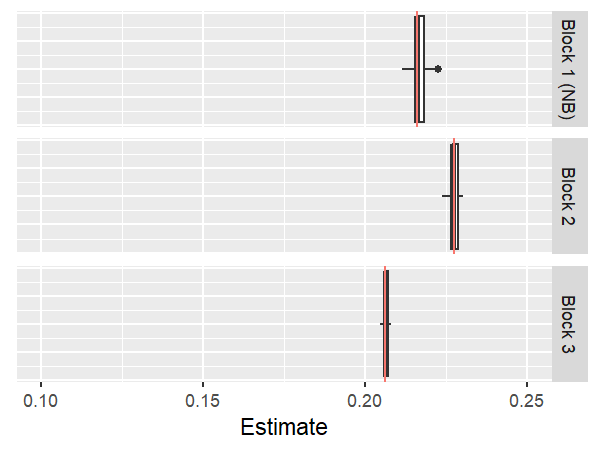}    \\
	\end{tabular}
	\begin{tabular}{c}
		${\rho}_{q}$ \\
		\includegraphics[width=0.4\linewidth]{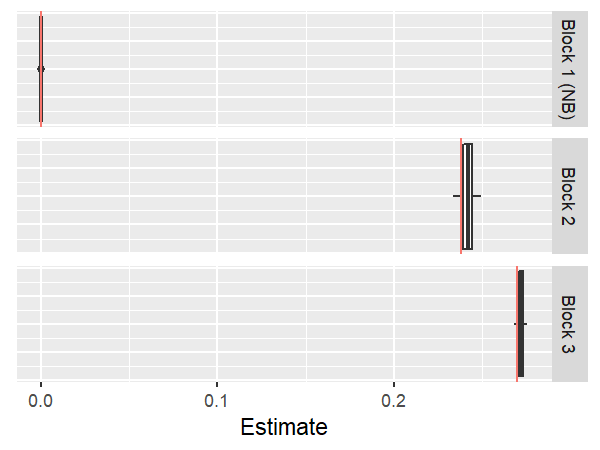} 
	\end{tabular}
	\caption{\footnotesize  \ch{\textit{Boxplots for repeated estimates of simulations. We ran the algorithm applied to 100 randomly generated networks with the same ground truth parameters and sample sizes but with variable group sizes. Each boxplot represents the summary of 50 individual estimates corresponding to 50 runs. The red bands represent the ground truth parameters for means, variances, and correlations. } }}
	\label{fig:simu_boxplots}
\end{figure}


The only `free' parameter in our algorithm is the \textit{number of blocks} specified $Q$.
\textit{Model selection} in the SBM clustering context usually refers to specified $Q$ during VEM estimation.
Existing approaches \cite{Daudin2008,mariadassou2010,matias2016} consider the \textit{integrated complete likelihood} (ICL) for assessing block model clustering performance \ch{in weighted as well as simple graphs}.  
We apply the method for a range of  $\widehat{Q}$ (as the \textit{estimate} for number of blocks). Results show that the usage of ICLs reaches its highest value at the correct ground truth value  of 3 and verify that this metric is suitable for evaluation of the method  (Figure \ref{fig:elbossim200} in Appendix \ref{app:simu_ICL}). 

%



\subsection{Method Comparisons for Synthetic Experiments}

We evaluated \textit{ARI}  (Adjusted Rand Index) and \textit{NMI} (Normalized Mutual Information) scores \cite{wilson_essc,palowitch_continuous,matias2016} for the six methods for these simulations.
The scores are both between 0 and 1 and serves as a proxy for \textit{percentage recovery}.
We have found that \texttt{SBANM} outperforms competing methods in every setting. We note that there is perhaps some implicit bias in favor of  \texttt{SBANM} because the simulations are generated according to the model. However, differences in recovery rates still elucidate some important information regarding the method's efficacy.

\toch{We present these results also as a part of Table \ref{table:method_compare}, earlier use for method comparison for PNC data.}
Table \ref{table:method_compare} (left) shows the recovery rates of the  NMIs and ARIs for all of the runs of \texttt{SBANM} are 1, suggesting perfect recovery. Recovery rates of the spectral clustering for both the naive (single-graph sum) and multigraph spectral clustering results are around half, while \textit{MASE} and \texttt{dynsbm} yield much better results, of up to 98\% agreement.
Results for other methods suggest effective \textit{partial} recovery of the memberships from \texttt{dynsbm} and \textit{MASE} even if none the network block structures are not perfectly recovered. 
None of the competing methods perfectly recover the block structures for the multilayer networks. 


Computation times are also assessed for the various methods. The spectral methods (rows 2, 4, 5) of Table \ref{table:method_compare} (left) have a nearly instant run time. \texttt{SBANM} and \texttt{dynsbm} take longer to compute. Under the (correct) specification for $Q=3$  \texttt{SBANM} , the computing time averages around 1100 seconds (in CPU time), while 350 for $Q=2$ and 3000 for 4 blocks.  
For  \texttt{dynsbm} the total procedure took about 3500 seconds in CPU time (with an optimal 8 clusters), which is slightly less than running \texttt{SBANM} for 2, 3, and 4 clusters. As such, our proposed method is fairly slow, but comparable to existing methods.



\section{Discussion}  \label{sec:disc} 
\ch{
	We have introduced a novel method \texttt{SBANM} that is  motivated by real-world clinical phenomenon of psychosis progression.
	\texttt{SBANM} is an unsupervised data-driven approach to identify groups of psychopathologies  that describe patterns of prodromal subjects transitioning to psychosis spectrum.
	Future developments of method could potentially lead to a deeper understanding of the transition from prodrome to psychosis spectrum and finally to schizophrenia using statistical network theory.  
	We demonstrated the relative benefits of this model compared to existing methods.
}

Network data  come in complex forms. They are particularly synchronous with the surge in data availability.
Our primary contribution was to introduce the notion of structured noise to weighted multilayer SBMs as well as an algorithm to estimate it.
Other work has explored cases where between-block transitions are all uniquely parameterized \cite{matias2016}, but they do not account for correlations between layers nor do they separate signal from noise. The proposed model is parsimonious and reveals more interpretable results that could be useful in clinical settings (more details in Appendix \ref{app:parsimony}).
In practice, $NB$ does not necessarily represent  a control group in PNC early adult males but  rather a dynamic cluster that specifically captures and 
\ch{reflects} the most notable interactions.

We have demonstrated that the  method is able to uncover latent, non-trivial patterns in a psychiatric condition (as well other data in Appendices \ref{app:voting},\ref{app:results_bikeshare}).
\ch{In Section \ref{sec:results_clinical}, we have shown that \texttt{SBANM} reveal a moderately sized group in the male early adult age subset with high rates of psychosis spectrum incidence and, perhaps more importantly, independence across different prodromal psychopathologies. These findings are useful for the study of psychosis in its ability to separate subjects with independent prodromal conditions from those with co-occurring ones.} \ch{The results from the applications to psychopathology data concurs with the ongoing discourse in moving away from  \textit{nosology} where psychiatric disorders are treated as discrete entities as opposed to multifaceted pathologies \cite{nosology}}. Etiologically, the proposed methodology supports the shift away from one-dimensional causal assumptions and instead to multifaceted casual pathways underlying severe psychiatric disorders such as psychosis and schizophrenia.

There are some limitations to \texttt{SBANM}. The issue of computation time persistently plagues variational methods. The algorithm slows when $K$ or $Q$ is large. \ch{However, its unique ability to partitioning independent from correlated clusters may be of more importance than speedy computation form a less nuanced method}.  SVI speeds up computation time to to make possible what was previously infeasible. Future work  may further explore subsampling methods with faster computation times.

\textit{Ambient noise} in networks  are related to overlapping  SBMs.  Some community detection methods adhere to a \textit{bottom-up} heuristic where clusters increase in size until memberships become stable; and naturally allows for separation between \textit{signal} and \textit{noise}. Many of these approaches implicitly assume  inherent structure but do not assign an explicitly parametric model to signal or noise \cite{wilson_essc,bodwin2015testingbased,palowitch_continuous}. Members not assigned to communities are  called \textit{background} nodes are identified but not statistically modeled.
As in many subdomains of statistics and signal processing, noise plays a large role in network theory and methods.
\ch{Existing work that discusses noise in networks  \cite{Zhang_Noise_2015,Priebe_2019,Miao2021InformativeCI}, 
	mostly describe it theoretically, however, few authors specifically address noise in the context of SBMs and fewer seek to explicitly model it for practical purposes.
	Though our approach theoretically isolates ambient noise in that it is uncorrelated,  in this study the \textit{noise block} actually yields the most meaning.}

The development of \texttt{SBANM} opens up many methodological avenues. One immediate next step is to expand the study of PNC data to include neuroimaging and genomics data. Such work is currently in progress for the PNC study to identify and jointly model neural and genetic influences in addition to symptoms. Another direction is in assessing significance or predictive power of the clusters. More generally, these in-group and out-of-group interactions are related to mixed effects models for multimodal weighted networks that may serve as another perspective in the study \ch{of} longitudinal analysis of networks \cite{Snijders05modelsfor}.

\section*{ Reproducibility}
Code and sample data for \texttt{SBANM} is available at \url{https://github.com/markhe1111/SBANM}.

\section*{Acknowledgements and Funding Information}
This project was funded by the Rockefeller University Heilbrunn Family Center for Research Nursing (RX, 2019) through the generosity of the Heilbrunn Family. The funding organizations had no role in the design and conduct of the study; collection, management, analysis, and interpretation of the data; preparation, review, or approval of the manuscript; and decision to submit the manuscript for publication.
MH was supported by the NSDEG fellowship.

Philadelphia Neurodevelopment Cohort (PNC) clinical phenotype data used for the anal- yses described in this manuscript were obtained from dbGaP at \url{http://www.ncbi.nlm. nih.gov/sites/entrez?db=gap} through dbGaP accession \url{phs000607.v3.p2.}
Support for the collection of the data for Philadelphia Neurodevelopment Cohort (PNC) was provided by grant RC2MH089983 awarded to Raquel Gur and RC2MH089924 awarded to Hakon Hakonarson. 

The authors thank Andrew Nobel and Shankar Bhamidi for helpful coments and theoretical advice. In particular, we thank them for defining and recognizing the problem of differential, correlated communities amongst multilayer networks. We also thank Professor Galen Reeves for helpful advice in contextualizing this work to the literature.  


\section*{Supplementary Material for  Community Detection in Weighted Multilayer Networks with Ambient Noise}
\begin{appendix}

	
	\section{Proofs and Derivations} \label{appendix:proofs_and_deriv}

	In this appendix we provide the proofs and derivations for the terms for the algorithm updates in Section \ref{sec:est_algorithm}. 
	Note that in Appdenices \textbf{A}, \textbf{B} and \textbf{C} we use slightly different notation: instead of $\mathcal X$ and $\mathcal Z$ for data and membership matrices, we revert to that used by other work and instead use $\Xb$ and \textbf{Z}. Note that this notation interferes with the layers of the network in  in the PNC case study, but for these first three appendices they are representative of more general cases.

	\subsection {Proof for Hierarchical ELBO} \label{app:ELBO_Proof}
	This is a proof paraphrased from Ranganath et al. \cite{ranganath2016} that the hierarchical ELBO is a sharper lower bound than the ELBO.  
	An inequality can be drawn between  the ``ordinary" ELBO $ \mathcal{L} $ without any hierarchical information and the Hierarchical ELBO 
	\begin{align*}
		\mathcal{L} 
		& = \Expec_{ R_{\text{hv}} (\Zb)  } [       \log \jointPxz]  + \mathcal{H}_\text{hv} ( R(\Zb )) \\
		& \geq  \Expec_{ R (\Zb, \Cb  )}[       \log \jointPxz]     - \Expec_{R ( \Zb, \Cb ) }\left[    \log R (\Zb, \Cb) \right]   + \Expec_{R (\Zb, \Cb ) }  \left[  \log S( \Cb | \Zb) \right] 
		\\
		& := \mathcal{L}^\prime  (\text{Hierarchical ELBO }).
	\end{align*}  
	The inequality in the above relationship arises from the  decomposition of the entropy $\mathcal{H}_\text{hv}$ of the  hierarchical distribution.  The proof of the inequality is based on the proof from Ranganath et al. \cite{ranganath2016} :
	\begin{proposition}
		\begin{equation*} 
			\mathcal{H}_\text{hv} ( R(\Zb )) 
			\geq   - \Expec_{R ( \Zb, \Cb ) }\left[    \log R (\Zb, \Cb) \right]   + \Expec_{R (\Zb, \Cb ) }  \left[  \log S( \Cb | \Zb) \right]. 
		\end{equation*} 
	\end{proposition}
	
	\begin{proof}
		\begin{align*}
			\mathcal{H}_{\hv} ( R(\Zb )) & =  - \Expec_{R_{\hv} (\Zb )}[\log R_{\hv} (\Zb )]  \\
			& = - \Expec_{R_{\hv} (\Zb )}  \left[  \log R_{\hv} (\Zb )   - \KL \left(  \Rcz ( \Cb | \Zb );   \Rcz (\Cb | \Zb ) \right)    \right] 
			\\
			&\geq  - \Expec_{R_{\hv} (\Zb )}  \left[  \log R_{\hv} (\Zb ) + \KL \left(   \Rcz (\Cb | \Zb )  ;   S ( \Cb | \Zb ) \right) \right] 
			\\
			& =-\Expec_{R_{\hv} } \left[  \Expec_{R (\Zb )}[\log R_{\hv} (\Zb )]   + \log \Rcz (\Cb | \Zb )    -    \log S (\Cb| \Zb )  \right]  \\
			& =  -\Expec_{R( \Zb, \Cb ) }\left[    \log R_{\hv} (\Zb) + \log R_{\Cb|\Zb } (\Cb | \Zb )  - \log S( \Cb | \Zb )  \right] \\
			& =  -\Expec_{R( \Zb, \Cb ) }\left[    \log R_{ \Zb, \Cb } (\Zb, \Cb ) - \log S(\Cb | \Zb )  \right] 
		\end{align*} 
	\end{proof}

	\subsection{Preservation of  $ \Expec_{ R (\Zb, \Cb  )}  [\log f (\Zb ) ] $ }  \label{app:preserve_ElogfZ}
	Here we show that the term for $   \Expec_{ R (\Zb, \Cb  )}  [\log f (\Zb ) ] $ as written in Eq.  \eqref{eq:fZb_same} is the same as in prior studies such as Daudin et al. \cite{Daudin2008}  
	\begin{proposition}
		$$	\Expec_{ R (\Zb, \Cb  )}  [\log f (\Zb ) ] = \sum_{i,q}       \tau_{iq} \log \alpha_q$$
	\end{proposition} 
	\begin{proof}
		
		\begin{align*}
			\Expec_{ R (\Zb, \Cb  )}  [\log f (\Zb ) ]  &= 
			\sum_i    \sum_q   \bigg(     P_q \tau_{iq} \log \alpha_q  + (1- P_q) \tau_{iq} \log \alpha_q  \bigg)  
			\\
			& = %
			\sum_q  (   P_q   + (1- P_q) )   \bigg(  \sum_i    \tau_{iq} \log \alpha_q    \bigg) 
			\\
			& =    \sum_{i,q}       \tau_{iq} \log \alpha_q   . 
		\end{align*} 
	\end{proof}
	
	\subsection{Derivation for Expected Log Likelihood} \label{app:derivationExpectedlogL}
	
	Description of the form of the joint likelihood  in Equation \eqref{eq:ElogfXZ} in Section \ref{sec:param_estimation}:
	\begin{proposition}
		The expected log likelihood of the multivariate normal distribution $   f (\Zb ,\Cb)   $  with respect to   $ R (\Zb, \Cb  ) $  is written as 
		$$ \Expec_{ R (\Zb, \Cb  )}[\log f (\Xb ,\Zb)]  = \ERhv [\log f (\Xb | \Zb) ] +  \sum_i    \sum_q     \tau_{iq} \log \alpha_q      $$ 
	\end{proposition}
	\begin{proof}
		\begin{align*}
			\Expec_{ R (\Zb, \Cb  )}[\log f (\Xb ,\Zb)] & =  \ERhv [\log f (\Xb | \Zb) ]
			+  \ERhv [\log f (\Zb ) ] 
			\\
			& = \ERhv [\log f (\Xb | \Zb) ] + \log  \Psi \sum_i    \sum_q    P_q \tau_{iq} \log \alpha_q   
			\\
			& \ \ + \log (1- \Psi) \sum_i    \sum_q  (1- P_q) \tau_{iq} \log \alpha_q  \bigg)  
			\\ & =  \ERhv [\log f (\Xb | \Zb ) ] + \sum_{i,q}  \bigg( P_q \tau_{iq} \log \alpha_q  +(1- P_q) \tau_{iq} \log \alpha_q  \bigg)  
			\\ & =  \ERhv [\log f (\Xb | \Zb ) ] + \sum_{i,q}     \tau_{iq} \log \alpha_q   
			\stepcounter{equation}\tag{\theequation}\label{eq:HELBO}  
		\end{align*}
	\end{proof}
	
	\subsection{Derivation of Joint Distribution $\Expec_{ R (\Zb, \Cb  )}[\log R (\Zb, \Cb  ) ] $}  \label{app:deriv_JointDist}
	
	The expectation of the log of the joint variational distribution is as follows:
	\begin{align*}
		\Expec_{ R (\Zb, \Cb  )}[\log R (\Zb, \Cb  ) ] 
		&=  \Expec_{ R (\Zb, \Cb  )}[ \log   R (\Zb | \Cb ) ]+  \Expec_{ R (\Zb, \Cb  )}[ \log   R (\Cb ) ]
		\\
		&=  \sum_{i,q}  \bigg(  (1-P_q) \Expec Z_{iq}   \log(\tau_{iq}  )   + P_{q}   \Expec Z_{iq} \log(\tau_{iq}  ) \bigg) + \Expec[ \log   R (\Cb ) ]
		\\
		& = \sum_{i,q} \tau_{iq} \log \tau_{iq} +   \sum_q    \bigg( P_q  \log P_q   +  (1-P_q)         \log  (1-P_q)    \bigg) 
	\end{align*}

	\section{ELBO and Hierarchical ELBO} \label{app:ELBO}
	This section contains details for the hierarchical ELBO as well as  the derivations for these expressions. The definition of ELBO is as follows:
	
	\begin{definition} {(Evidence Lower Bound (ELBO))}  \label{def:ELBO}
		Given observed data $\Xb$  with unknown latent membership variables $\Zb $, 	the evidence lower bound (ELBO) 
		$ \mathcal{L} $ is the 	approximately  optimal likelihood that minimizes the KL Divergence between the approximate distribution $R( \Zb , \Cb )$  and the  posterior frequency $f(   \Zb , \Cb |   \Xb ) $. It is expressed as follows:	
		\begin{align*} 
			\mathcal{L}  
			&= \Expec_{ R_{\text{hv}} (\Zb)  }    \left[      \log \jointPxz - \log R_{\hv} (\Zb ) \right] 
		\end{align*} 
		Alternatively, the ELBO can be rewritten as the sum of the expected 
		frequency and the entropy $\mathcal{H}$ of variational variable $\Zb$:
		\begin{align*} 
			\mathcal{L}   = \Expec_{ R_{\text{hv}} (\Zb)  } [       \log \jointPxz]  + \mathcal{H}_{\hv}( R(\Zb)).
		\end{align*} 
	\end{definition}  
	
	In Appendix \ref{app:ELBO_Proof} we have shown  the inequality between  the ``ordinary" ELBO $
	\mathcal{L} $  and Hierarchical ELBO $\mathcal{L}^\prime  $. We write  $ \mathcal{L}^\prime $ here as follows:
	
	$$ \mathcal{L}^\prime  = \Expec_{ R (\Zb, \Cb  )}[       \log \jointPxz]     - \Expec_{R ( \Zb, \Cb ) }\left[    \log R (\Zb, \Cb) \right]   + \Expec_{R (\Zb, \Cb ) }  \left[  \log S( \Cb | \Zb) \right] $$

	\subsection{ Log likelihood Part of Hierarchical ELBO} \label{app:ElogfX|Z}
	The log likelihood portion of the hierarchical ELBO is written as : 
	\begin{align*} 
		& \Expec_{R_{\bX}} [\log( f   (  \Xb| \Zb  )  ) ] 
		= 
		\sum_{q }    P_q  \sum_{i }    \sum_{j }  \    \tau_{iq} \tau_{jq} \bigg( \frac 1 2   ( \Xb_{ij}     - \boldsymbol{\mu}_{q}   )^T  \boldsymbol{\Sigma}_q^{-1} ( \Xb_{ij}      - \boldsymbol{\mu}_{q}   )         -   (2\pi)^{K/2}  (\log    |\boldsymbol{ \Sigma}_q |  )^{  1 /2 }
		\bigg)    
		\\
		& +   \sum_q (1 - P_q)  \sum_{i }    \sum_{j }     \tau_{iq} \tau_{jq} \bigg(  \frac 1 2   (  \Xb_{ij}      - \boldsymbol{\mu}_{AN}   )^T  \boldsymbol{\Sigma}_{AN}^{-1} (  \Xb_{ij}      - \boldsymbol{\mu}_{AN}   )         -  (2\pi)^{K/2}  (\log    | \boldsymbol{ \Sigma} _{AN}|^{  1 /2 }
		\bigg)    \\
		& +   \sum_{q  }   \sum_{l: l \ne q   }    \sum_i  \sum_j  \    \tau_{iq} \tau_{jl} \bigg(  \frac 1 2   (  \Xb_{ij}      - \boldsymbol{\mu}_{AN}   )^T  \boldsymbol{\Sigma}_{AN}^{-1} (  \Xb_{ij}      - \boldsymbol{\mu}_{AN}   )         -  (2\pi)^{K/2}  (\log    | \boldsymbol{ \Sigma} _{AN}|^{  1 /2 }
		\bigg).    
	\end{align*}

	\subsection{Expression for Hierarchical ELBO}\label{app:full_ELBO}
	The full form of the hierarchical ELBO is the log likelihood part (Section \ref{app:ElogfX|Z}) plus the membership probabilities, entropy, and their hierarchical counterparts:
	\begin{align*} 
		\mathcal{L} ^\prime 
		=  \Expec_{R_{\bX}} [\log( f   (  \Xb| \Zb  )  ) ] 
		& +  \sum_{i,q}    \tau_{iq} \log \alpha_q    
		- \sum_{q}  \sum_i  \tau_{iq} \log \tau_{iq} -    
		\\   \sum_q    \bigg( P_q  \log P_q     + & (1-P_q)   
		\log  (1-P_q)    \bigg)  +   \sum_i \sum_q \bigg( P_q \log \Psi +(1- P_q) \log (1-\Psi)  \bigg)  \tau_{iq}
	\end{align*}
	This is the full expression for the hierarchical ELBO as described in Section \ref{sec:ELBO_Optim}.

	\section{Calculations for Variational EM Algorithm} \label{app:VEM_Calc}
	
	This section gives derivations for every step of the Variational EM algorithm in Section \ref{sec:est_algorithm}.
	
	\subsection{Optimizing Membership Probabilities $\boldsymbol{\tau}$ in E-Step}\label{app:tau_est}
	We find optimal values for each $\tau_{iq}$ by solving this following equation, which is described in Section
	\ref{sec:E_stepTauEst}:
	\begin{align*} 
		\frac{\partial }{\partial \tau_{iq}}    \mathcal{L}  =& 
		\log(\alpha_q )
		+  \sum_{k \le K }   \sum_{j \le n  }   \tau_{jl}    \bigg( P_q f(  X^k_{ij}, \boldsymbol{\mu}_q, \boldsymbol{\Sigma}_q )   + (1-P_q)  f(X^k_{ij}, \boldsymbol{\mu}_{AN}, \boldsymbol{\Sigma}_{AN} )   \\
		& + \sum_{l \le Q: l \ne q }  f(X^k_{ij}, \boldsymbol{\mu}_{AN}, \boldsymbol{\Sigma}_{AN} )   \bigg) -      \log(\tau_{iq} )-1  
		+    P_q \log \Psi +(1- P_q) \log (1-\Psi) 
		\\  
		:=&0 , 
	\end{align*}
	rearranging $\tau_{iq}$ we solve this equation using a fixed point iteration procedure

	\subsection{Estimation of Noise Probability $P_q$ in E-Step} \label{app:noise_probPq}
	Variational variables $P_q$ that serve as the ``soft" versions of  $C_q$ 
	can be  approximated by estimating the probability of block $q$ being a ``signal" block or noise block.  
	The terms $ \ERhv \log f (\Xb | \Zb ) ,    \Expec[ \log   R (\Cb ) ]$,  $\Expec[ \log   S (\Cb| \Zb  ) ] $ in $\mathcal{L}^\prime $ are dependent on $\Cb $. 
	Practically, because we need to  normalize for $N_q$, which is $1-P_q$, that variable is more simple (if not the only possible tractable option).   
	
	\begin{align*} 
		\frac{\partial }{  \partial  N_q  }  \mathcal{L}^\prime 
		=&  \frac{\partial }{  \partial  N_q  }  \ERhv [\log f (\Xb | \Zb) ] 
		-      \log N_q +  \log (1-N_q)     - ( \log \Psi + \log (1-\Psi)   ) \sum_i \tau_{iq}
		\\
		:=& 0    
	\end{align*}  
	where the first term is $f(\cdot )$ is the portion of the multivariate normal density
	as described in Section \ref{app:ElogfX|Z}.
	\begin{align*} 
		\sum_{k  }  \sum_{i,j} \tau_{iq} \tau_{jq}  \bigg(     f(X^k_{ij}, \boldsymbol{\mu}_q, \boldsymbol{\Sigma}_q ) - f(X^k_{ij}, \boldsymbol{\mu}_{AN}, \boldsymbol{\Sigma}_{AN} )   +   \log  \bigg(
		\frac{ 1 - \Psi }{ \Psi }
		\bigg)   \bigg)   
		& =   \log  \bigg( \frac{ N_q }{ 1-N_q } \bigg) 
	\end{align*}
	So then, after rearranging:    
	\begin{align*} 
		\widehat{N_q} =   \bigg(   1+   \bigg[  \exp  \bigg(   \sum_{k  }  \sum_i \sum_j \tau_{iq} \tau_{jq}  \bigg(     f(X^k_{ij}, \boldsymbol{\mu}_q, \boldsymbol{\Sigma}_q ) - f(X^k_{ij}, \boldsymbol{\mu}_{AN}, \boldsymbol{\Sigma}_{AN} )  )   +   \log  \bigg(
		\frac{1-\Psi }{ \Psi }
		\bigg)   \bigg)  \bigg) \bigg]  ^{-1}   \bigg) ^{-1}. 
	\end{align*}  
	Then the final $N_q$ estimates are made after normalizing all $\widehat{N_q} $ such that they sum to one. Finally, the $P_q$ estimates are made by subtracting $N_q$ from 1.

	\subsection{Derivation of  Signal Terms for  M-Step} \label{app:M_StepSignalTerms}	
	The closed-form estimate of the parameter for the mean vector $   {\boldsymbol{ \mu } }_{q} $ for each block $q$   from the M-step is
	\begin{align*}
		\widehat{\boldsymbol{ \mu } }_{q} & = \frac{ \sum_{i,j} \tau_{iq} \tau_{jq} \Xb_{ij}    }{ \sum_{i,j} \tau_{iq} \tau_{jq}    }   P_q    +  \frac{ \sum_{i,j} \tau_{iq} \tau_{jq}  \boldsymbol{\mu}_{AN}   }{ \sum_{i,j} \tau_{iq} \tau_{jq}    } \cdot  (1- P_q )
		\\ & = \frac{ \sum_{i,j} \tau_{iq} \tau_{jq} \Xb_{ij}     }{ \sum_{i,j} \tau_{iq} \tau_{jq}    }   P_q +     \boldsymbol{ \mu}_{AN}    (1- P_q )
	\end{align*}
	
	Assuming convergence of $P_q$ to either 0 or 1 within the context of the variational iterations, the theoretical value of 
	\begin{align*}
		\boldsymbol {\mu}_{q} &=
		\begin{cases*} 
			\frac{ \sum_{i,j} \tau_{iq} \tau_{jq} \Xb_{ij}    }{ \sum_{i,j} \tau_{iq} \tau_{jq}    }  
			& if $q$ is Signal: $P_q=1$
			\\
			\boldsymbol{\mu}_{AN}   & if $q$ is Noise: $P_q=0$
		\end{cases*}
	\end{align*}
	
	Similarly to mean calculations, the variance calculations (along diagonals) are  : 
	\begin{align*}
		\widehat {  \boldsymbol{\Sigma}_q}  &  =   \frac{  \sum_{i,j}   \tau_{iq}  \tau_{jq}  ( \Xb_{ij} - \boldsymbol \mu_q  )^2     } {  \sum_{i,j}   \tau_{iq}  \tau_{jq}    }\cdot   P_q   
		+  \boldsymbol \Sigma_{{AN}} \cdot (1-P_q ) 
		\\
		&=
		\begin{cases*}
			{ \sum_{i,j} \tau_{iq} \tau_{jq} (\Xb_{ij} - \boldsymbol  \mu_q )^2    }   \big/ { \sum_{i,j} \tau_{iq} \tau_{jq}     }   & if $q$ is Signal: $P_q=1$
			\\
			\boldsymbol \Sigma_{AN} & if $q$ is Noise: $P_q=0$
		\end{cases*}
	\end{align*}
	
	The  cross-term for two layers  $h, k$  is written as:
	\begin{align*} 
		{ \widehat{\boldsymbol  {\Sigma}_{hk,q}}}
		& = \frac{   \sum_{i,j}   \tau_{iq}  \tau_{jq}  ( \Xb_{k,ij} - \boldsymbol \mu_{q,k}) ( \Xb^h_{ij} - \boldsymbol \mu_{q,h})     } {  \sum_{i,j}   \tau_{iq}  \tau_{jq}   } \cdot  P_q  + 0 \cdot (1-P_q)
		\\
		&=\frac{   \sum_{i,j}   \tau_{iq}  \tau_{jq}  (\Xb^k_{ij} - \boldsymbol \mu_{q,k}) ((\Xb^h_{ij} -  \boldsymbol \mu_{q,h}  )         } {  \sum_{i,j}   \tau_{iq}  \tau_{jq}   } \cdot P_q
	\end{align*}
	The element-wise  correlations at iteration $t$ across layers  $h,k$ ($h \ne k$)  are then calculated as 
	\begin{align*}
		\hat{\rho_q}^{h,k}  & = \frac{\widehat{ {\Sigma_{hk}^q}}}{    
			\sqrt{\widehat{ {\Sigma^{ h}_q}} \widehat{ {\Sigma_{ k}^q}}} 
		}. 
	\end{align*}
	Finally,   the putative correlation (across all layers) for block $q$ is
	\begin{align*}
		\hat{\rho_q}  =   \max_{h,k}       \hat{\rho_q}^{h,k} .
	\end{align*}

	\subsection{Derivation for $\boldsymbol{\mu}_{AN}$ and $ \boldsymbol{\Sigma}_{AN} $} \label{app:deriv_muAN}
	
	This is derivation for \eqref{eq:mu_AN}
	To calculate the global parameters, the global noise probability term $\Psi$ defined previously is 
	\begin{align*}
		{ \widehat{   \boldsymbol{\mu}_{AN}    }}
		& = \Expec_{R_\Xb(\Zb, \Cb )}  \big[  \boldsymbol{\mu}_{AN}     \big] 
		\\ 
		& =     \Prob(B_q \ne  NB   )  \Expec_{R(\Zb, \Cb )} \big[  \boldsymbol{\mu}_{AN} \big| B_q \text{ is not } NB   \big]  
		\\ & \quad  +    \Prob(B_q = NB )     \Expec_{R(\Zb, \Cb )} \big[  \boldsymbol{\mu}_{AN} \big| \     \{B_q=NB\}\big];  \quad   q: 1 \le q \le Q
		\\
		& =    \Psi  \frac {  \sum_{j , i  }    \sum_{l, q: q \ne l } \tau_{iq}  \tau_{jl}   
			\Xb_{ij}  
		} {  \sum_{j , i  }    \sum_{l, q: q \ne l }     \tau_{iq}  \tau_{jl}         } 
		+(1-\Psi)
		\frac{	  \sum_{j , i  }    \sum_{ q }    \tau_{iq}    \tau_{jq}   (1-P_q)
			\Xb_{ij}  }{   \sum_{j , i  }    \sum_{ q }    \tau_{iq}    \tau_{jq}     (1-P_q) },
	\end{align*} 
	$ \widehat{   \boldsymbol{\Sigma}_{AN}    }$ can also be calculated in a similar way.
	
	\subsection{Derivation of $\Psi $} \label{app:derivPsi}
	In this section we derive $\Psi$ Let  $\{NB\} $ represent the event that there exists a Noise Block in the multilayer graph system. The we write the indicator for this event as $\textbf{1} (NB)$ with probability 
	$\Prob(NB)$. 
	\begin{align*}
		\Psi  &=  \Prob(B_q \ne  NB ;  \forall q:  q \le Q) 
		\\
		&= \Prob( C_q =1 ; \forall q:  q \le Q  ) 
		\\
		&=  1 - \Prob(   \text{Global average rate of } q \text{ s.t. } C_q =0; \forall q:  q \le Q)  \\
		& = 1 - 1/Q 
		\\
		& = (Q-1)/Q 
	\end{align*} 
	
	

	\section{Stochastic Variational Inference} \label{app:StochVI}.
	To speed up computation, we  apply  stochastic variational inference (SVI) to calculate the membership parameters $\tau_{iq}$ and $P_q$.
	We subsample nodes at each step of the E-step in variational EM. Calculating $\tau_{iq,t}$ and $P_{q,t}$ comprise two stochastic sub-steps of the E-step at iteration step $t$; we label their  SVI estimates as $\widehat \tau_{iq,t}$ and $\widehat P_{q,t}$. 
	At each $t$, we sample  a set of nodes $M = \{i_1,...,i_m\} $ of size $m$ and their associated edges from graph layers $\Xb^1,...,\Xb^K$.
	Let $\tau^m_{iq,t }$ represent the randomly subsampled graph at iteration step  $t$.
	\begin{enumerate}
		\item   (Calculating $\tau^m_{iq,t}$)
		Partial updating step for $\tau^*_{iq,t }$ at time $t$  wherein the subsampled memberships $i,j \in M$ are found:   	
		\begin{align*} 
			\tau^*_{iq, t}   \propto  \exp \bigg(   &\log(\alpha_q ) +  \sum_{k \le K}   \sum_{  j, l \in M   }   \tau_{jl  , t-1 }   \bigg(   P_q f(X^k_{ij}, \boldsymbol{\mu}_q, \boldsymbol{\Sigma}_q )   + (1-P_q)  f(X^k_{ij}, \boldsymbol{\mu}_{AN}, \boldsymbol{\Sigma}_{AN} ) \\
			&+       \sum_{l : l \ne q }  f(X^k_{ij}, \boldsymbol{\mu}_{AN}, \boldsymbol{\Sigma}_{AN} )    \bigg)   
			-1  + P_q \log \Psi +(1- P_q) \log (1-\Psi) \bigg).  
		\end{align*}  
		The update step averages the newly calculated $  \tau^*_{iq, t} $ with the previous value
		\begin{align*}
			\widehat{ \tau}_{iq, t} =   
			\delta_t   \tau^*_{iq, t}  +  (1-\delta_t) \widehat  \tau_{iq, t-1} .
		\end{align*} 
		\item   (Calculating $P_{q,t}$)
		The signal probability 
		$P_q$ is calculated in \eqref{eq:Pq_calculation} but with the same subsampled replacements as done in the previous calculation of $\boldsymbol{\tau}$. For each time point the new noise probability $p^*_{q,t}$ is calculated and  averaged with the previous noise probability at time $t-1$. The update step is 
		\begin{align*}
			\widehat P_{q,t} =  
			\delta_t P^*_{q,t}  +  (1-\delta_t)\widehat P_{q,t-1}    .
		\end{align*} 
	\end{enumerate}

	\subsection{Details on Stochastic Variational Inference} \label{app:StochVI}
	To apply stochastic variational inference, we first define the time-variable \textit{subsampling parameter}  $\delta_t$  to retain some memory from previous iteration. 
	At every step $t$, a subsampled index set $B( )\delta_t) \in [n]$   is randomly drawn from the data, then the step of the algorithm is only applied to the subsample $\Xb_{B(\delta_t)}   $.
	A time-varying  $\delta_t \in (0,1)$ is selected to satisfy the convexity assumption of 
	(1) $\sum_t \delta_t = \infty $ and (2)  $\sum_t \delta^2_t < \infty $ as outlined in \cite{hoffman2012stochastic}, for some $\kappa \in (.5,1)$
	\begin{align*}
		\delta_t &= (t +1)^{-\kappa }.
	\end{align*} 
	
	However, this criteria needs to be changed when the stochastically sampled variables represent memberships. Empirically, the samples converge at a fast rate when the initial ``burn in" steps are subsampled, with subsample sizes increasing with each successive step. If subsampling does not take place, a potentially major impediment may arise from the slow computation speed in early steps where initialized estimates are not near the optimal values. As such, the step sizes are set as such: 
	\begin{align*}
		\delta_t &= \min \bigg(  a +  \bigg( \frac{t} {t +1} \bigg)^\kappa n  , n   \bigg).
	\end{align*} 
	$a$ and $\kappa$ are constants. $a$ governs the smallest subsample size and $\kappa>1$ governs the rate of increase for subsample size at each step size, with the maximum possible subsample size $n$. A larger $a$ means a larger starting subsample, and a larger $\kappa$ means a faster rate of increase in subsample size. 
	
	Empirically, for a wide range of simulations, an effective choice for $a$ is between 100 to 200 (depending on network size) and for $\kappa$ is 2.
	These values are chosen to ensure computational efficiency in addition to accuracy: computation times for initial values are much slower if the parameter estimates are far from the optimal values which maximize the ELBO, so smaller sample sizes in earlier iterations will economize computation  by producing more local minima, while later iterations will yield more globally accurate estimates \cite{hoffman2012stochastic}.
	

	\section{Identifiability and Parsimony} \label{app:connection} 
	
	\subsection{Identifiability and Connection to Prior Models} \label{app:identifiability} 
	In the introduction, we reference the \textit{affiliation model} in Section \ref{sec:contributions_context} as an example of  prior work describing global noise on networks. On a single weighted network, a simple parametric model   known as the affiliation model described in Allman et al. \cite{Allman_2011} is formulated as follows with piecewise global  fixed rates:
	\begin{align*}
		\mu_{ql} &= (1 -     	p_{ql} )       \delta_0 + p_{ql} F_{ql} \big( \theta_{\text{in}} \textbf{1}_{q = l}  +  \theta_{\text{out}} \textbf{1}_{q \ne l}  \big); \quad  1 \le   q,l \le Q
	\end{align*}
	where probability $p_{ql}  $ is the sparsity parameter, continuous distribution  $F_{ql} (\theta_{ql}) $ with parameter $ \theta_{ql}   $
	and $\delta_0$ is a dirac mass at zero, and with probability
	\begin{align*} 
		p_{ql}    &=
		\alpha  \textbf{1} _{q = l} + 	\beta \textbf{1} _{q \ne l}; \quad . 
	\end{align*}
	One can conceive of the weighted stochastic blockmodel as a special case of the \textit{general form of mixture models for random graphs} described in \cite{Allman_2011}. For graph $X$ where each weighted edge is $X_{ij}$ between nodes $i,j$:
	\begin{align*}
		\forall q,l \in \{1,...,Q\}   \quad     X_{ij} | \{  Z_{iq }Z_{jl}=1  \} \sim p_{ql}   f(  \cdot, \theta_{ql} )
		+ (1- p_{ql}  )  \delta_0 (\cdot ),
	\end{align*}
	where $p_{ql}$ serves as the sparsity parameter between 0 and 1, which represents the proportion of . 
	$f(  \cdot, \theta_{ql} )$ represents the parametric family of distributions at specified in group-interactions $q$ and $l$.
	The conditional distribution of $X_{ij}$ is a mixture of the Dirac distribution at zero representing non-present edges. 
	The proposed \texttt{SBANM} model can also be viewed as an instance of the  generalized model above.    It is a   mixture of the affiliation model and the weighted multilayer SBM.  
	Matias et al. \cite{matias2016} discuss identifiability of block parameters in multilayer SBMs. The authors cite  \cite{Allman_2011} in setting the conditions for identifiability for weighted SBMs over multiple layers.   Since the affiliation model is also proven to be identifiable \cite{Allman_2009}, we posit that \texttt{SBANM}  is also identifiable.
	
	In practice, the membership and parameter recovery in simulations in Section \ref{sec:expeirmental_procedure} suggests that the model is identifiable empirically. However, theoretical justifications may be pursued in future work.

	\subsection{Parsimony Compared to Other Models} \label{app:parsimony}
	\texttt{SBANM} is a parsimonious compared to most other models. If inter-block interactions ($B_q \ne B_l$) are all unique, as in some models \cite{matias2016,mariadassou2010}  then this lends to overparametrization, especially at high dimensions ($\approx K \times \frac{ Q(Q-1)}{2}$ parameters). The number of parameters may be reasonable for binary and Poisson-distributed multilayer networks, but will quickly inflate in the multivariate Gaussian case.
	\texttt{SBANM} yields $2KQ + Q-1  + 2K$ parameters comprising the $2KQ$ mean and (diagonal elements of) variance parameters $\{  (\boldsymbol{  \mu  }_q,  \boldsymbol{  \Sigma }_q )  \}_{q:  q \le Q}$ , $Q-1$ correlation parameters $ \{ \rho_q\}_{q:  q \le Q, q \ne q_{NB}}$, and $2K $ noise parameters $(\boldsymbol{  \mu  }_{AN},  \boldsymbol{  \Sigma }_{AN}   ) $. As $Q$ becomes large, the number of parameters increases quadratically in the canonical weighted SBM but linearly in  \texttt{SBANM}. As $K$ becomes large, also, the rate of increase for parameters in the proposed method is smaller than that in existing methods. This advantage is demonstrated in computing time comparisons in Section 6.3.1.
	

	\section{Additional Simulations}
	
	In this section, we describe three additional simulations that were conducted for the proposed method. The experiment applies the proposed method on simulations with varying parameters. The second experiment applies the method to networks generated form the same parameters. The third experiment verifies the usage of  the \textit{integrated composite likelihood} \cite{matias2016} selecting the optimal $Q^*$. Finally, the last section runs the method on some larger networks.

\subsection{Simulations of Networks with Differing Parameters (Experiment 1)} \label{appsec:simu_experiment1}

We first describe the simulation scheme of  the first experiment. The means for each unique block for every network are randomly generated from a Gaussian distribution centered around 0 and 2 respectively for the first and second layers. 
After the parameters are generated,   the observations are simulated from multinormal distributions governed by these parameters.
Each network has $AN$ governing both a single block $NB$ and interstitial noise $IN$ that is centered around (-1,0).  
We repeat this procedure for trivariate networks of $n=200$ nodes, wherein the Gaussian priors for each (signal) block have means of -2, 0, and 2 respectively for the first, second, and third layers.  
In order to ensure the separability of blocks during simulations, we only select the networks whose blocks' minimum Bhattacharya distances are above a certain threshold.  We calculate the minimum Bhattacharya distances between blocks across 500 simulated networks, and then select the networks with the largest 10\% of the minimum Bhattacharya distances to filter out the networks whose blocks are `far enough away' from each other; we run 50 instances of the \texttt{SBANM} algorithm for both the bivariate  ($n=500$) and the trivariate case ($n=200$).

\subsubsection*{Results}
Fifty runs of the algorithm were performed for both the bivariate and trivariate networks with differing parameters. 500 networks were generated as described in the previous section, then networks with the highest 10\% of the minimum Bhattacharya Distances between clusters' parameters are retained. 

Though this experiment is primarily focused on \textit{membership recovery}, parameter estimation remains as a byproduct.
Across many simulations with a variety of parameters, there does not seem to be much systemic bias in the estimates as empirical means of differences between estimated and true parameters are centered around 0.
Median percentage differences, across all estimated parameters, between the estimates and true values are between 20 to 25\% for bivariate, and 10-20\% for trivariate networks. Histograms for the mean and variance parameters (each distinct parameter is treated like an observation) show essentially matching distributions between estimates and ground truth parameters for means (\ref{fig:simu_histos}).

A slight discrepancy between distributions for variance parameters ($\sigma^2_{q,k}$ for $k = 1,2,3$)  among trivariate networks. This slight bias  may be related again to the curse of dimensionality and, while does not seem to elicit too severe a problem in the clustering results, may be investigated in future endeavors.

Percentage differences between the estimated and ground-truth parameters also show moderately accurate recovery in both bivariate and trivariate networks. The lowest 25\% quartiles for all parameters are between 0 and 3 percent and show that these estimates are very close to the ground truths.  Conflated with the relatively higher mean and median differences, the low 1st quartiles show that accuracy for parameter runs seem to occur along a binary: either estimates are very close to their targets, 
or they are fairly far off. Some of the high percentage differences may arise from small ground-truth values, which are divided to calculate percentage differences. Others may arise from the mismatches in clustering memberships.  However, this limitation mostly arises in the trivariate case, as there is a near-perfect recovery rate for the bivariate simulations. 

\begin{figure} [htbp]
	\centering 
	\begin{tabular}{ccc}
		\hline
		\multicolumn{3}{c}{\textsc{  Histograms of True and Estimated Parameters}}  \\ 
		\hline 
		& \textbf{Bivariate}  & \textbf{Trivariate} \\
		\hline
		\\
		{\rotatebox[origin=l]{90}{  	\quad \quad 	\quad \quad 	  \footnotesize      Estimated\slash True 	$\mu_{k,q}$}}  
		&  
		\includegraphics[width=0.4\linewidth	]{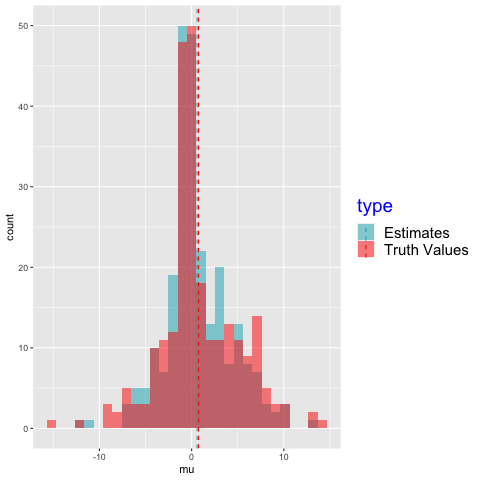}
		&   		\includegraphics[width=0.4\linewidth]{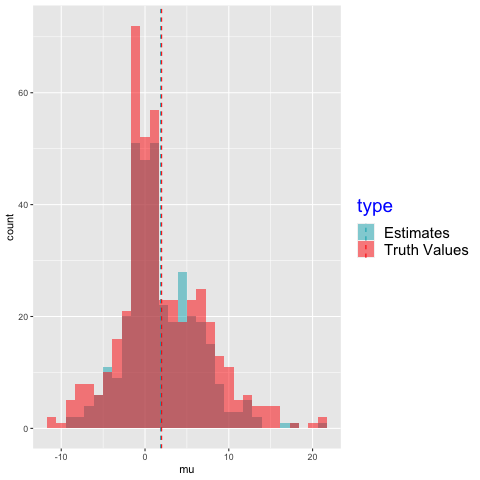} \\
		{\rotatebox[origin=l]{90}{ 		\quad \quad 	\quad \quad   \footnotesize       Estimated \slash True 	$\sigma^2_{k,q}$}}    &  	\includegraphics[width=0.4\linewidth]{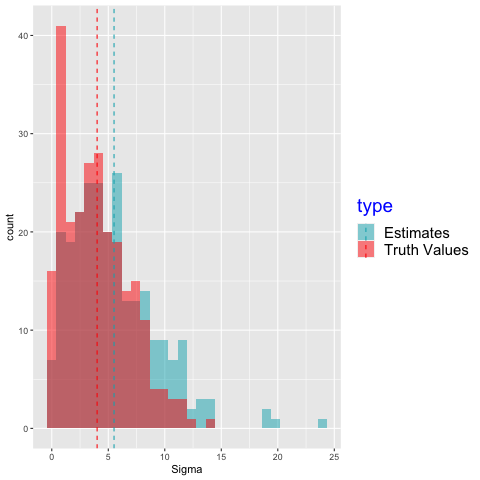}
		&   		\includegraphics[width=0.4\linewidth]{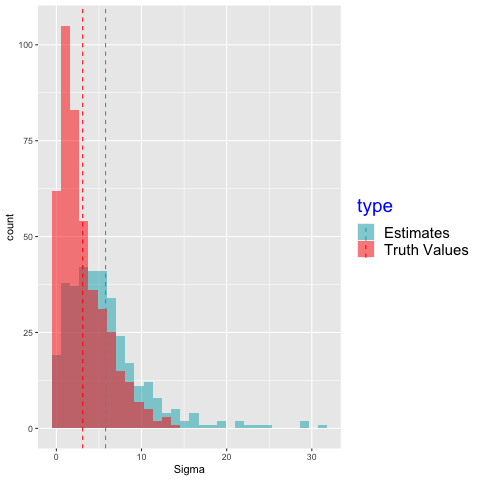} 
		\\
		\hline
	\end{tabular}
	\caption{\footnotesize Histograms of ground truth (red) and estimate (blue) parameter values for the 2-layer and 3-layer networks compared to the estimated parameters from the algorithm. Parameters across layers are all plotted together.
		Dashed lines demarcate the empirical means of these estimated and ground truth parameters. 
		For ground truths (red), these empirical means are 
		.75 for $\mu_{k,q}$ (bivariate, top left), 
		1.98 for $\mu_{k,q}$ (trivariate, top right),
		4.01 for $\sigma^2_{k,q}$ (bivariate, bottom left), 
		3.10 for $\sigma^2_{k,q}$ (trivariate, bottom right).
		For estimates of parameters, they are
		.58 for $\mu_{k,q}$ (bivariate, top left), 
		1.84 for $\mu_{k,q}$ (trivariate, top right),
		5.51 for $\sigma^2_{k,q}$ (bivariate, bottom left), 
		5.58 for $\sigma^2_{k,q}$ (trivariate, bottom right).
	}
	\label{fig:simu_histos}
\end{figure} 

\subsection{Simulations of Networks with the Same Parameters (Experiment 2)} \label{appsec:simu_experiment2}

The first experiment was conducted primarily to demonstrated membership recovery under a variety of different parameters and block sizes. 
The purpose of the second experiment, which runs the algorithm under a set of fixed parameters, is to show that the method recovers parameters effectively.
The fixed parameters were generated through simulation with fixed Gaussian distributions with prior means 10,15, and 20 and prior variance parameters of 5. The first entries of each layer correspond to the noise block with fixed means at 5, 10 and 15. The means are:
$\boldsymbol{\mu}_{X,q} =    (5, 11.98, 11.55, 10.39)$,
$\boldsymbol{\mu}_{Y,q} =   (10 ,16.86, 16.49, 14.81)$,
$\boldsymbol{\mu}_{Z,q} = (15 ,16.69 ,21.25 ,21.08)$.
The variances are 
$\boldsymbol{\Sigma}_{X,q}$ = ( 7.88, 13.11, 0.31, 1.16), 
$\boldsymbol{\Sigma}_{Y,q}$=  ( 7.32,  7.67, 4.89, 1.03), 
$\boldsymbol{\Sigma}_{Z,q}$ =(6.69,  4.15, 0.06, 4.36).
The correlations are $\rho_q$ = (0.00, 0.40, 0.15, 0.34), 
and the true group sizes are 76 nodes for the first  block ($NB$), 97 for the second, 93 for the third, and 34 for the fourth.

\subsubsection*{Results}
We generated  100 networks following these exact specifications and ran \texttt{SBANM} on all of them. In Figure \ref{fig:simu_boxplots} in the main text, each boxplot comprises a set of 100 estimates for each parameter values. The first row shows those for the first layer (written as $\Xb$), the second $\Yb$, the third $\Zb$, and the fourth for correlations between the three layers. The red band shows the true parameter values as listed above.

\begin{figure}[htbp] 
	\centering
	\begin{tabular}{ccc}
		\hline
		\multicolumn{3}{c}{\textsc{  Boxplots of Estimated Parameters}}  \\ 
		\hline
		&	$\boldsymbol{\mu}_{k,q}$&	$\boldsymbol{\Sigma}_{k,qq}$ \\
		\rotatebox[origin=l]{90}
		{ 	\quad \quad 	\quad \quad 	\quad  \footnotesize    $\Xb$} &	\includegraphics[width=0.4\linewidth]{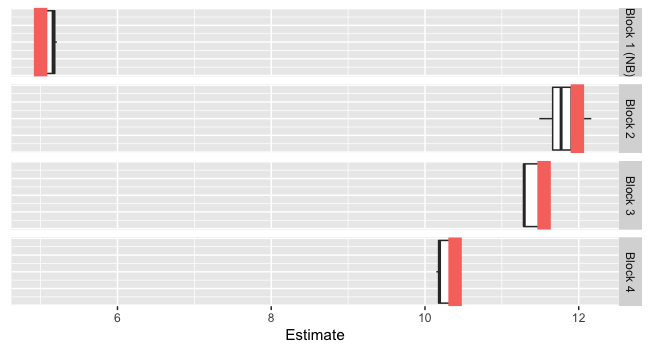}	& 	\includegraphics[width=0.4\linewidth]{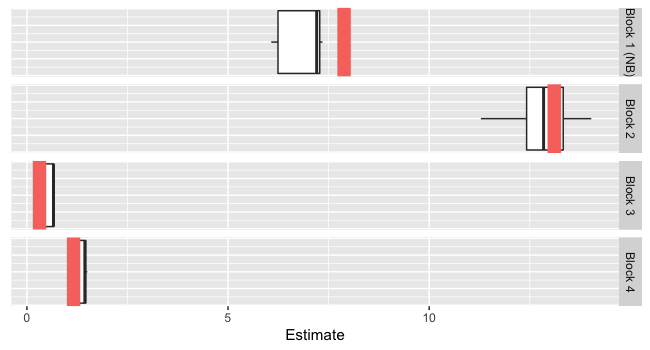} 	 	\\  
		\rotatebox[origin=l]{90}{ 	\quad \quad 	\quad \quad \quad	  \footnotesize    $\Yb$}
		& \includegraphics[width=0.4\linewidth]{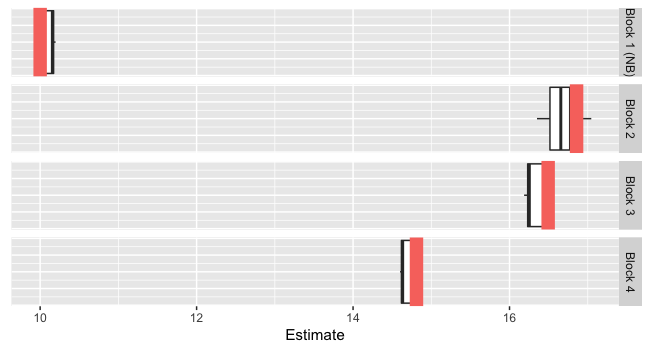}	&    \includegraphics[width=0.4\linewidth]{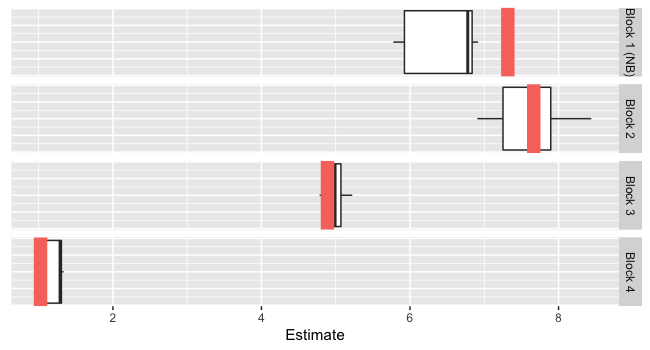} \\
		\rotatebox[origin=l]{90}{ 	\quad \quad 	\quad \quad 	\quad  \footnotesize    $\Zb$} 
		&\includegraphics[width=0.4\linewidth]{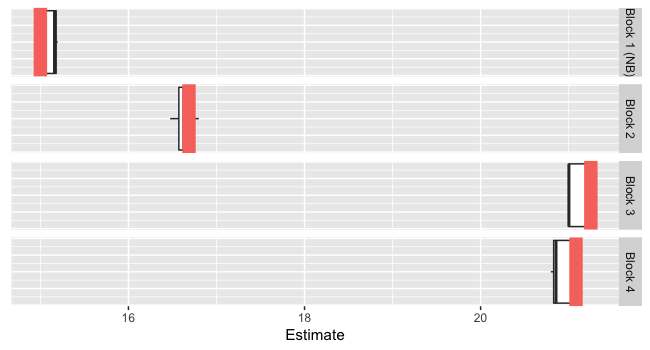}  
		&\includegraphics[width=0.4\linewidth]{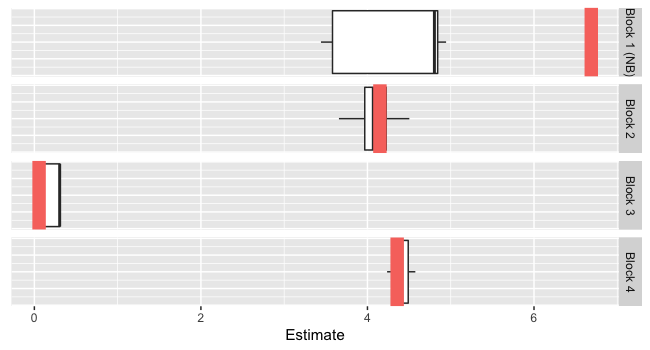}    \\
	\end{tabular}
	\begin{tabular}{c}
		${\rho}_{q}$ \\
		\includegraphics[width=0.55\linewidth]{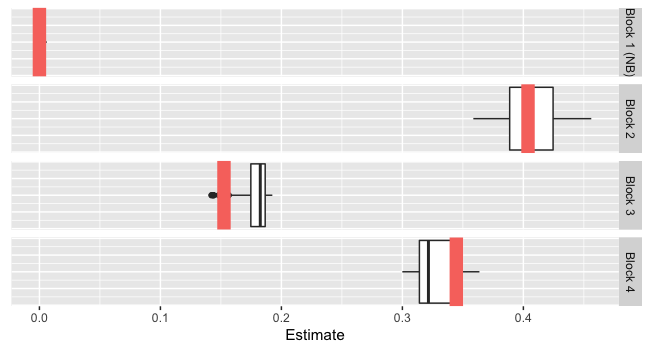} 
	\end{tabular}
	\caption{\footnotesize \textit{Boxplots for repeated estimates of simulations (second type). We ran the algorithm applied to 100 randomly generated networks with the same ground truth parameters and fixed sample sizes. Each boxplot represents the summary of 100 individual estimates corresponding to 100 runs. The red bands represent the ground truth parameters for measn, variances, and correlations.  }}
	\label{fig:simu_boxplots}
\end{figure}

\subsection{ICL Assessment (Experiment 3)} \label{app:simu_ICL}

\textit{Model selection} in the SBM clustering context usually refers to selection of the number of a priori blocks before VEM estimation as it is the only `free' parameter in the specification step of the algorithm.  
Existing approaches \cite{Daudin2008,mariadassou2010,matias2016} consider the \textit{integrated complete likelihood} (ICL) for assessing block model clustering performance. Matias et al.  write the ICL for multilayer graphs in the following way (adapted to match the notation of this study)
\begin{equation} \label{eq:ICL}
	ICL(Q) = \log f( \Xb , \Zb )  - \frac 1 2 Q(Q-1) \log ( n (K-1))  - pen(n,K, \boldsymbol{\Theta} )
\end{equation}
to translate the terminology, $\boldsymbol{\Theta} $ corresponds to the total set of transition parameters in the SBM, where  $\boldsymbol{\Theta} : =\boldsymbol{\Theta}_\text{Signal}  \bigcup \boldsymbol{\Theta}_\text{Noise} $ \cite{matias2016}. The penalty parameter $pen( \cdot )$ is chosen dependent on the distributions of the networks; the `Gaussian homoscedastic' case in Matias et al. is derived to be 
\begin{align*}
	pen(n,K, \boldsymbol{\Theta}) = Q\cdot \log \bigg( \frac{n(n-1) K }{2} \bigg) 
	+ \frac{Q(Q-1) }{2}   K \cdot \log \bigg(   \frac{n(n-1) }{2}  \bigg) .
\end{align*} 
Though the authors made  the assumptions that the variances are constant for all blocks, we assume that the models are similar enough to \texttt{SBANM} such that the evaluation criterion is applicable to our case. 
For this portion of the simulation experiment we fix $n$ at 200 and the ground-truth $Q$ at  5. However, we apply the method for a range of hypothesized block numbers $\widehat{Q}$ (as the \textit{estimate} for number of blocks) from 2 to 7. Simulation results show that the usage of ICLs caps at $\widehat{Q}=5$, the correct ground truth value (Figure \ref{fig:elbossim200}).

\noindent
\textbf{Results:} We used a single instance of a trivariate network with 200 nodes from the simulations generated in the first experiment (Section \ref{appsec:simu_experiment1}). ICLs for five runs of the algorithm were calculated. Each run presupposed a different selection of $Q$ from 2 to 7. The ground-truth value of $Q$ is 5 and Figure \ref{fig:elbossim200} showed that the ground-truth $Q$ captured the highest ICL.
\begin{figure} [htbp]
	\centering
	\includegraphics[width=0.52\linewidth,  trim= {0cm .8cm 0cm 2.3cm} , clip	]
	{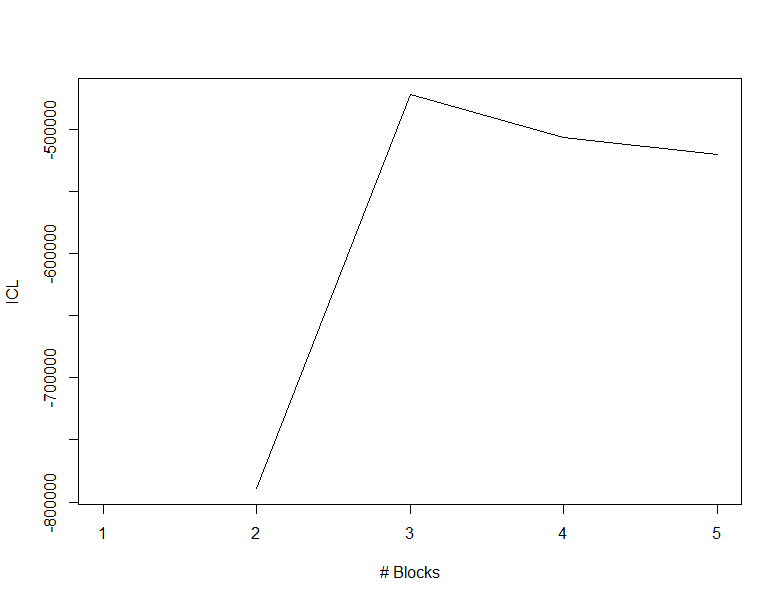}
	
	\caption{\footnotesize{} \textit{ICLs for simulation study for three-layer network of 200 nodes with a ground-truth $Q$ of 5, which maps to the maximum ICL that was found by the method of estimation.}}
	\label{fig:elbossim200}
\end{figure}

\subsection{Large Network Simulations}
For large-network simulations, single instances of networks with $n=1000$ and 2000 are generated for $Q=4$ and 5. Results yielded exact recovery for memberships and within 5\% errors for parameters.

	\section{Details for Analysis of PNC Data}  \label{app:PNC_Details}

	\subsection{PNC Preprocessing and Network Construction}\label{app:data_pnc_proc}
	The PNC has  a well-represented sample with youth of mostly European American ancestry but include a substantial portion of African Americans.
	Roughly 21\% met psychosis spectrum criteria and 4\% reported threshold psychosis symptoms
	(\cite{Calkins_PNC_Psych2017}).    We separately analyze the two age cohorts \textit{youth} (with sample size 5136) and \textit{early adult}  (sample size 1863).
	
	Response networks are constructed using a function that gauges similarity as well as positivity or negativity of responses. This distance function is similar to Hamming distance, but takes into account the direction of \textit{positive} or \textit{negative} agreement and is between -1 and 1 .  
	In a  single graph-layer $\Xb^ k$, a weight $ X_{ij}^k$ between two nodes is derived from indicators $h^k_{ij,u}$ across $U$ questions (indexed by $u$) pertaining to a given set of conditions. 
	\begin{align*}
		h^k_{ij,u} = \begin{cases}
			1& \text{if } i,j  \text{ both answer ``yes"}   \\
			-1& \text{if } i,j \text{ both answer ``no"}   \\
			0& \text{otherwise }    \\
		\end{cases}
	\end{align*} 
	Each $h^k_{ij,u}$ between two subjects $u,v$ is -1 if both answer no, 1 if both yes, otherwise 0.  These values  are then summed  and divided by the total number of questions $U$: $$ r^k_{ij} =  \frac{  \sum_{u = 1,..,U }      h^k_{ij,u}   }{U}   . $$ The weight $r^k_{ij}$ is 1 if two subjects both answer yes to everything and -1 if they answer no to everything. The weight $r^k_{ij}$ is then transformed using a Fisher transformation to produce a  value that approximates an observation in a normal distribution, in layer $k$:
	$X^k_{ij} = \text{Fisher} (r^k_{ij}).$

	\ch{\subsubsection{Exploratory Histograms}}
	We show exploratory histograms for the male early adult sample in the PNC data that serves as the primary sample for the analysis. Each edge (duplicates removed) show that each layer appears to be composed of a mixture of normal distributions.
	\begin{figure}[htbp]
		\centering
		\includegraphics[width=0.325\linewidth]{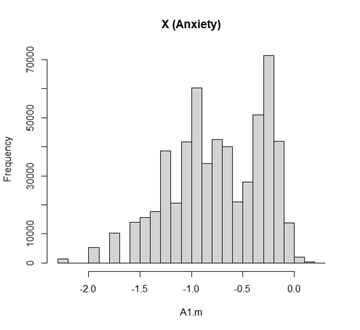}
		\includegraphics[width=0.325\linewidth]{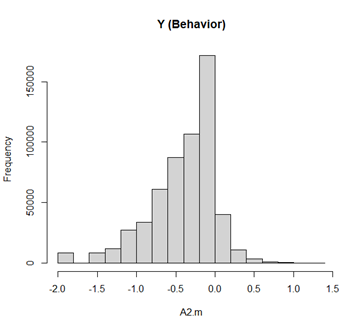}
		\includegraphics[width=0.325\linewidth]{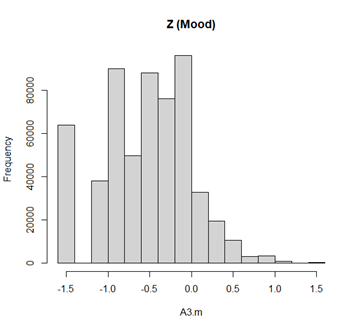}
		\caption{Histograms of the $\Xb, \Yb, \Zb$ layers representing the anxiety, behavior, and mood disorder layers. }
		\label{fig:fig0xhist}
	\end{figure}

	\ch{\subsection{Parameter Estimates and Demographic Characteristics}}

	\begin{table}[htbp]
		\centering
		\begin{tabular}{rrr|rrrrrrrr}
			\hline
			& block &  & $n$ & $\rho$ & $\mu_X$ & $\sigma_X$ & $\mu_Y$ & $\sigma_Y$ & $\mu_Z$ & $\sigma_Z$ \\ 
			\hline
			1 &  $NB$ &  & 41 & 0.00 & -0.63 & 0.31 & 0.02 & 0.39 & -0.22 & 0.46 \\ 
			2 &  $S_1$  & & 244 & 0.24 & -0.23 & 0.15 & -0.07 & 0.14 & -0.02 & 0.48 \\ 
			3 &  $S_2$ &  & 471 & 0.27 & -1.17 & 0.31 & -0.76 & 0.43 & -0.93 & 0.45 \\ 
			\hline
		\end{tabular}
		\caption{Parameter estimates of \texttt{SBANM} in male early adults in PNC data. $\rho$ represents the correlation, $\mu$s represent the means, and $\sigma$s represent the variance.} 
	\end{table}
	
	Demographic characteristics among the separated blocks do not elucidate any major differences in sex, age, or race among the different clustered blocks. The \textit{environmental impacts} variable, however, does show some major difference in the $NB$ cluster, perhaps this shows that there could be some causal environmental effects into the psychosis-prone subjects. Unfortunately, we do not have more information on how exactly these environmental factors may play a role in this.
	Future research may investigate deeper into this relationship.
	
	Demographic characteristics of the clustered subjects are shown in 	Table \ref{fig:ClusDemog_table}.
	Rates of patients who are African American, Hispanic, or female are roughly even across the board for most clusters for both youth and early adult under different $Q$ specifications.  
	
	\begin{table}[htbp]
		\centering
		\begin{tabular}{rrrr|rrrrr}
			\hline
			& Block &  & $n$ & Age & Env. & Black & Hispanic & Female\\ 
			\hline
			1 & $NB$ &  & 41 & 20 & 27 & 32\% & 5\%& 63\%\\ 
			2 & $S_1$&  & 244 & 20 & -10 & 37\% & 7\%& 53\% \\ 
			3 & $S_2$ & & 471 & 20 & -15 & 37\% & 6\% & 63\%  \\ 
			\hline
		\end{tabular}
		\caption{Demographic characteristics for \texttt{SBANM} in male early adults in PNC data. The variables represent age, environmental factors (Z-scores), percenatage black, 
			\textit{Demographic Characteristics of PNC Results. The columns represent respectively: age, environmental factors (Z-scores multiplied by 100),   \% African American, \% Hispanic, and \% Female}	\label{fig:ClusDemog_table}
		}
	\end{table}

	\subsection{Additional Posthoc PNC Analyses} \label{app:pnc_posthoc}
	Hypothesis tests between different imputed blocks in PNC psychopathological networks (post-processed)    and diagnostic categories showed significant differences between all the different clusters. In EA, though the diagnostic comparisons (right) are not all significantly different from each other, the signal (correlated) blocks are all signififcantly different from the noise block $NB$ at the significance level of 0.05.


\begin{table}[!htb]
	\caption{\footnotesize 
		Hypothesis tests for the clustered blocks in \textit{Youth} subjects along two different criteria. In the first assessment (left), edges in the weighted network for each layer are treated as a i.i.d sample and compared to other edges using  t-tests. In the second assessment, proportions of positive clinical diagnoses are tested across different imputed blocks. 
		Recall that $\textbf{X}$ represent the network of symptom response similarities for anxiety, $\textbf{Y}$ for behavior, and $\textbf{Z}$ for mood disorders.
	} \label{table:Youth_ClusSig}
	\begin{minipage}{.5\linewidth}
		\footnotesize
		\begin{tabular}{clll}
			\hline
			\hline
			\multicolumn{4}{c}  {\textsc{Edge Comparison for EA (3 Gps)}}\\
			\hline
			$B_q$ Comp. & $\textbf{X}$ & $\textbf{Y}$  & $\textbf{Z}$  \\ 
			\hline
			$NB$- $S_1$& 0.00 (**)& 0.00 (**)& 0.00 (**)\\ 
			$S_1$-$S_2$& 0.00 (**)& 0.00 (**)& 0.00 (**)\\ 
			$NB$-$S_2$& 0.00 (**)& 0.00 (**)& 0.00 (**)\\ 
			\hline
		\end{tabular}
	\end{minipage}%
	\begin{minipage}{.5\linewidth}
		\centering
		\footnotesize
		\begin{tabular}{lllll}
			\hline
			\hline
			\multicolumn{5}{c}  {\textsc{Diagnosis Comparison for EA (3 Gps)}}\\
			\hline
			\%Anx & \%Beh & \%Mood & \%TD & \%Psy \\ 
			\hline
			0.00 (**)& 0.00 (**)& 0.00 (**)& 0.00 (**)& 0.00 (**)\\ 
			0.00 (**)& 0.00 (**)& 0.00 (**)& 0.00 & 0.00 (**) \\ 
			0.00 (**)& 0.03 (**)& 0.00 (**)& 0.00 (**)& 0.00 (**)\\ 
			\hline
		\end{tabular}
	\end{minipage} 
\end{table}


	\section{Analysis of US Congressional Voting} \label{app:voting}
	The focus of the study is on the PNC data. However, we also show the model's generality by applying  the method to political and human mobility data.  We use \texttt{SBANM} to find latent patterns in longitudinal US congressional co-voting data to analyze the static as well as dynamic patterns in co-voting amongst US congressional districts, historically a fruitful domain of network analysis \cite{cho2011coevolution}.
	We also find clusters in longitudinal aggregations of bikeshare networks, whose stations are represented by nodes. Analysis of zones  amongst urban mobility services is elucidating for discovering latent patterns within human geography and demographic trends \cite{he2020intertemporal,carlen2019role,cazabet_using_2017}. 
	
	
	In the voteview data, each layer represents interactions among each congressional session. $(\Xb,\Yb)$  represents the 100th and 115th sessions of congress, respectively. $n$  represents the number of congressional seats that are common to all three sessions (new or relabeled seats that were added since the first session are not included)
	Only two layers are used for this application of \texttt{SBANM} to the Divvy data,  and $(\Xb, \Yb)$ in this case represents the normalized, aggregated trips between 2014-2016 and 2016-2018 respectively. The sample size $n=547$  describes the total number of stations and each edge weight represents aggregate trips between stations.

	We use congressional voting records  from \textit{Voteview}  to uncover patterns in US congressional voting patterns that may yield more nuanced political groups than party labels (i.e. Democrat, Republican) over time.
	We use a similar pre-processing step as done for the PNC data to assign measures for co-voting similarities between seats in the US House of Representatives during the  100th, and 115th sessions. Voting similarities between representatives in Congress are represented as weighted edges between nodes (representing members). 
	Each layer corresponds to a different congressional session.  
	We apply the proposed model to data from the  \textit{Divvy} bikeshare system  in Chicago called to show the ways that
	demarcating zones of bikeshare trips change across different years. 
	Trip data for Divvy are publicly available on their respective websites \cite{divvy_data}.  
	
	The overarching motivation for this application is belied by the assumption that political parties change over time and do not necessarily capture the political ``tribes'' in the US House of Representatives in the past and the present.  
	Prior work use co-voting patterns in the congress and senate in the United States to demonstrate applications of multilayer SBMs by representing district representatives (or senators) as nodes and their covoting similarities as edges \cite{wilson_temporal,cho2011coevolution}. 
	Though most congressional seats have fixed political parties that are representative of their political alignments, parties are assemblages of many constituents with issues that often fragment or congeal (ie polarize) over time. 
	As such, it is useful to trace and segment the groups that either vote with each other persistently, or change drastically following some  shift. 
	Clustering different political `tribes' by their similarities in voting is important for studying and forecasting patterns in US politics.
	It particular, it may be of interest to look for certain ``swing" districts that yield more signal for political analysts to study, compared to the ambient levels of connectivity in politically non-contentious districts.
	
	We procure voting data from \textit{Voteview} \cite{voteview}. We  use data from all congressional line items from the 100th (1987-89), and 115th (2017-19) sessions, excluding consensus votes where all votes were `yes' or `no'. These sessions sample distinct decadal political milieus in the United States across 30 years and serve as snapshots indicating long-term changes in the political inclinations of congressional districts. Though the number of these districts total 435 presently, differing seats often appear and vanish due to redistricting, and we use the seats that were common to both sessions. The resulting network size $n$ is 393.
	
	We use similarity measures similar to that which was applied to PNC survey data for voting records. Between two district seats, which are represented by nodes $i$ and $j$, the total votes in agreement (both yes or both no)  are summed, then subtracted by the total disagreeing votes and divided by the total votes cast.  We convert this correlation-like value, which is between -1 and 1, to a statistic that  approximates to a normal distribution by applying the same Fisher transformation used in Section \ref{sec:data}.
	Like in other studies \cite{wilson_temporal}, consensus votes that have either 100\% ``yes" or 100 \% ``no" are omitted.
	
	We ran the algorithm over a range of values for estimated block numbers $  Q$, as was done in Section \ref{sec:expeirmental_procedure}. As the block sizes increase, the ICL also increases, until  $ Q:=3$ which is where it appears to attain a maximum. 
	\begin{figure}[htbp]
		\centering
		\includegraphics[width=.52\linewidth,
		trim= {0cm .8cm 0cm 2.3cm} , clip	]{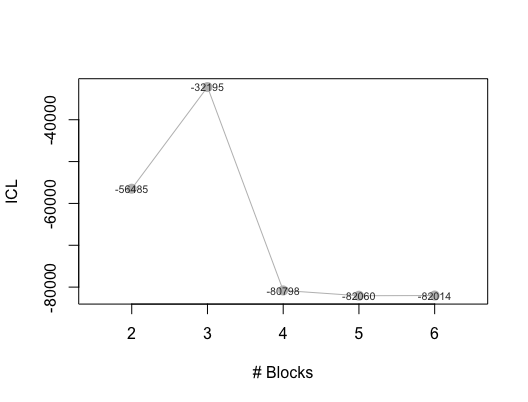}
		\caption{\footnotesize  \textit{Block selection for US congressional voting data based on the method; 3 blocks yields the greatest ICL.} }
		\label{fig:ptx37}
	\end{figure}
	We display the clustering results for 3 blocks are shown in Figure \ref{fig:ptx37}. In addition to the block sizes and estimated correlations, we show the average percentage of Republican party membership (\%\textit{R}) in the 100th and  115th sessions. The results show capture distinct shifts in party membership across the years: $NB$ appears to capture the 
	moderate niche of the congress.  
	
	\begin{table}[ht]
		\centering
		\begin{tabular}{ cc|ccc |ccl } 
			\hline
			\multicolumn{8}{c}  {\textsc{Memberships, Parameters, and Party Affiliation}}\\
			\hline
			\textbf{Block}	 & $n$ &$\mu_{X,q}$&$\mu_{Y,q}$& $\rho_q$& \footnotesize \%\textit{R(100{th})} &\footnotesize  \%\textit{R(115th)} &  \footnotesize  {Notable People}   \\ 
			\hline
			$NB$  & 9 & 0.02 & 0.31 & 0.00 & 36 & 67& \footnotesize Nancy Pelosi (1) \\ 
			$S_1$  & 233 & 0.71 & 0.36 & 0.09 & 4 & 50 &  \footnotesize Beto O'Rourke(2), Paul Ryan(2) \\ 
			$S_2$  & 151 & 0.55 & 0.45 & 0.04 & 99 & 68  & \footnotesize Dick Cheney(1), Liz Cheney(2)  \\ 
			\hline
		\end{tabular}
		\caption{\footnotesize \textit{ Clustering results for congressional voting data in the 100th and 115th sessions. In addition to the means and correlations of the (normalized) similarity networks, mean (Republican) party membership rates and notable people in each block are given.}}
	\end{table}
	
	Nine members in $NB$ vote at the same rate with each other  as with any other cluster; 
	The interpretation of this block as \textit{moderate} is supported by membership of \textit{moderate Democrat} politicians such as Nancy Pelosi who occupied the seat during empirically verified by the fact that  more than half of the block is Republicans in the 115th session.
	Moreover, $NB$ yields the same rate as every other block votes at the same rate with a different block. 
	
	The two biggest political enclaves are  large bipartisan \textit{party} that is half Democrat and half Republican in 2015 but was almost entirely Democrat in 1987 ($S_1$), and another group that was almost entirely Republican in 1987 but only about 2/3 Republican in more recent times. 
	The asymmetry in the blocks $S_1$ and $S_2$ is perhaps of note; one can view possibly $S_2$ as analogous to $S_1$, but more likely the block is capturing an uneven relationship where there is no Democratic equivalent to the Republican block  $S_2$ which shows entrenchment of voting ideology along geographical (district-wise) lines. These dynamics may be due to fundamental differences in voting patterns between the two parties.
	Results reveal the large drop-off in the Democrats' political dominance  in the 100th session.    
	Instead of capturing static (same-period) blocks, \texttt{SBANM} is able to capture some of the largest \textit{differential} movements between the 1980s and 2015.
	

	\section{Human Mobility  Data Analysis} \label{app:results_bikeshare}
	
	The \texttt{SBANM} method is applicable to human mobility patterns which is represented by bikeshare data. Bikeshare networks have been argued to trace the latent patterns within human mobility in urban systems \cite{cazabet_using_2017}. He et al. \cite{he2020intertemporal} and others have modeled bikeshare stations as nodes and aggregate trips as edges 
	\cite{carlen2019role}, and then gathered conclusions about the patterns of human mobility within these bike-sharing constraints. In particular, prior work have analyzed differences in time-of-day patterns, functional differences (ie work-to-home and home-to-home trips), as well as long-term usage between neighborhoods. 
	Carlen et al. have proposed a time-dependent SBM for (binary) paths between bikeshare stations
	\cite{carlen2019role}.
	%
	We convert trip data from the public records of the \textit{Divvy} bikeshare system into time-series networks where each edge represents trips and each node represents stations.
	We write these  network time-series as $\{G_s \}_{1 \le s \le S}$,  where $S$ is the aggregate weekly time-points between January 2014 to  June 2016 , and  $\{G_t \}_{1 \le t \le T} $  for $T$ as the aggregate weekly time-points between July 2016 to December 2018, as was done a previous analysis of the \textit{Divvy} system as conducted in He at al. \cite{He2019}.
	New stations  as well as stations that were removed during this time are omitted, such that the total number of stations $(n=547)$ is consistent across time. 
	
	We sum all of the edges across all time points for distinct time-periods $S$ and $T$  The two graphs $\Xb$ and $\Yb$ represent differential layers across two temporal regimes. 
	We use the number of  aggregated trips across  each time-regime $\Xb$ and $\Yb$ to represent edge-weights.
	The edge-weights are then transformed by dividing each value by the respective strengths (sum of weights) to procure a ratio between 0 and 1. The ratio is then converted  into an approximately normal value by the \textit{logit} transformation. Because of this transformation, mean values are negative and between -10 and -20. Estimated statistics (Figure \ref{fig:DivvyComm}) are reconverted using the inverse logit transform, then multiplied by the total graphwise sum-of-strengths, to convey a normalized mean rate of trips across stations within the same community.
	
	\begin{figure}[htbp]
		\begin{tabular}{cc}
			\boxed{ 	\includegraphics[width=0.42\linewidth, trim= {9cm 2cm 8cm 0cm}, clip]{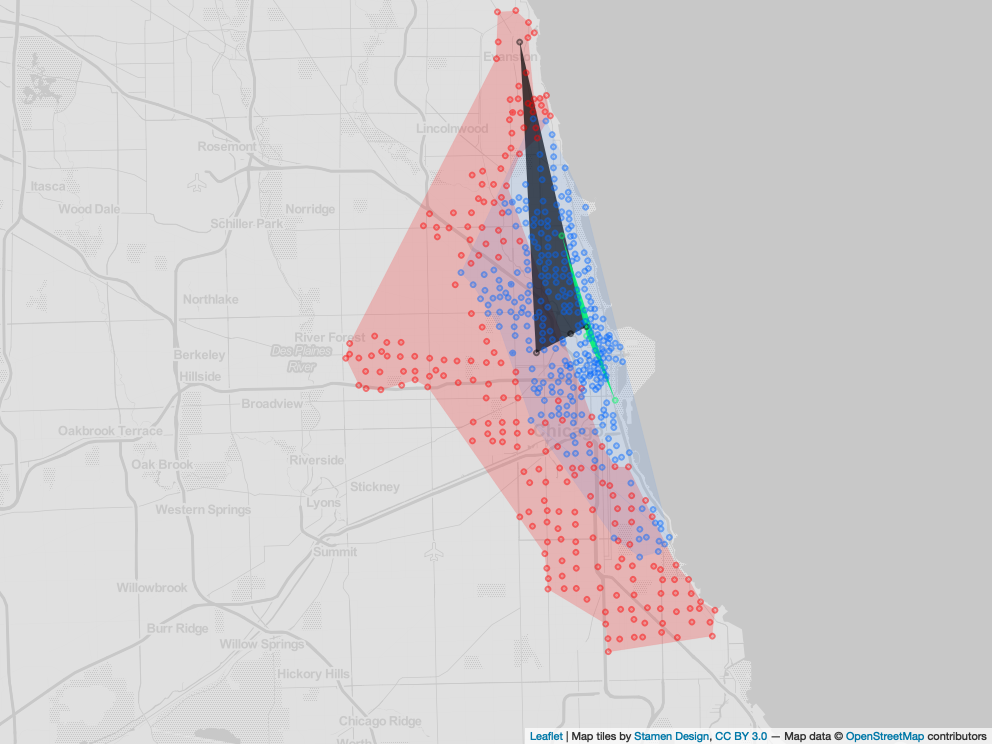} }	&
			\includegraphics[width=0.73\linewidth,
			trim= {8.1cm 14.4cm 4.5cm 16.3cm} , clip]{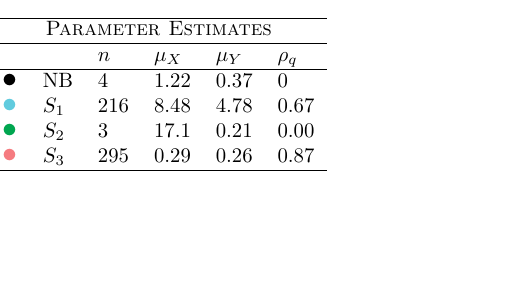}
		\end{tabular}
		\caption{\footnotesize \textit{Communities found across 2 time-periods in the \textit{Divvy} Bikeshare networks in Chicago, with associated (normalized) estimates for (normalized) mean rates of trips within the cluster in each time  period, as well as correlations.}
		}
		\label{fig:DivvyComm}
	\end{figure}
	
	Results show distinct geographical patterns (Figure \ref{fig:DivvyComm}).  The red cluster is  the largest (at 295 nodes) and represents a distinct baseline group for both time periods with activity that persist across time. The high inter-block correlation of .87 in this block suggests persistent trip interactions across time.
	The blue cluster represents a smaller (216 nodes) but a more \textit{persistent} area of activity: it has higher means for both the first and second layers than that of $S_1$ for both time-regimes, and also has a high  correlation rate. Because this area is closer to more affluent areas around the lake with more parklike amenities (such as the lakefront bike path), this block signifies zones with higher trip activity across both time periods.

	Smaller groups $NB$ and $S_2$ concentrate around the northern part of the city  and have very different estimated means that signal drastic change in usage over time. 
	Indeed, the green block $S_2$ has the highest first-layer mean $\mu_X$ but the lowest second layer mean $\mu_Y$. That the correlation in this block across layers is zero furthermore suggests  a disjointingly decreased usage over the two time periods.
	$NB$ is represented by the grey-black cluster in the northwest part of the city and has the same parameters of ridership as riders traversing across different blocks; which offers an interpretation to the large, but not infeasible, distance between stations (members) in this block. These discovered clusters have interpretable results and suggests tha viability of the method to human mobility data, after the appropriate transformations.
\end{appendix}

\bibliographystyle{imsart-nameyear} %


\end{document}